\documentclass{amsart}
\pdfoutput=1
\usepackage[utf8]{inputenc}
\usepackage[T1]{fontenc}
\usepackage{amsfonts}
\usepackage{amssymb}
\usepackage{mathtools}
\usepackage{charter}
\newcommand{\coloneqq}{=}
\newcommand*{\mytitle}{Sampling of stochastic operators}
\newcommand*{\myauthor}{G\"otz E. Pfander, Pavel Zheltov}
\usepackage{float, subfig}   
\usepackage[perpage,symbol*]{footmisc}  
\DefineFNsymbols*{mylamport}[math]{\dagger\ddagger\S\P\|{**}{\dagger\dagger}{\ddagger\ddagger}}
\setfnsymbol{mylamport}
\usepackage{aliascnt, verbatim}
\usepackage[section]{placeins}
\usepackage{array}
\usepackage{paralist,verbatim,lineno,scalefnt,xspace, chngcntr,extdash,booktabs,multirow,rotating,changepage,csquotes}
\usepackage[numbers]{natbib}
\usepackage[loose,nice]{units}
\usepackage[kerning=true]{microtype}
\usepackage{hyperref}
\hypersetup{  colorlinks,
  citecolor=blue,
  linkcolor=blue,
  urlcolor=blue,
pdfkeywords = {sampling, identification, stochastic, reconstruction, sounding, radar, channel sounding, spreading function, scattering function, time-varying operators},
pdfauthor = {\myauthor},
pdftitle = {\mytitle},
pdfpagelayout = OneColumn,
pdfstartview = FitH
}
\usepackage[OT2,OT1]{fontenc}
\DeclareSymbolFont{cyrletters}{OT2}{wncyr}{m}{n}
\DeclareMathSymbol{\Shah}{\mathalpha}{cyrletters}{"78}
\theoremstyle{plain}
\newtheorem{theorem}{Theorem}

\newaliascnt{lemma}{theorem}
\newtheorem{lemma}[lemma]{Lemma}
\aliascntresetthe{lemma}

\newaliascnt{proposition}{theorem}
\newtheorem{proposition}[proposition]{Proposition}
\aliascntresetthe{proposition}

\counterwithout{figure}{section}

\theoremstyle{definition}
\newaliascnt{definition}{theorem}
\newtheorem{definition}[definition]{Definition}
\aliascntresetthe{definition}

 \newtheorem*{remark}{Remark}

\newcommand*{\R}{\mathbb{R}}

\providecommand{\Z}{\mathbb{Z}}
\newcommand*{\C}{\mathbb{C}}
\newcommand*{\Expectation}{\mathbb{E}}
\newcommand*{\EXP}{\Expectation}

\renewcommand{\S}{\mathcal{S}} \newcommand*{\FT}{\mathcal{F}}

\renewcommand{\d}{\:\mathrm{d}}

\renewcommand{\phi}{\varphi}
\newcommand*{\eps}{\varepsilon} \renewcommand*{\emptyset}{\varnothing}

\DeclarePairedDelimiter\floor{\lfloor}{\rfloor}
\DeclarePairedDelimiter\abs{\lvert}{\rvert}
\DeclarePairedDelimiter\norm{\lVert}{\rVert}
\DeclarePairedDelimiter\absq{\lvert}{\rvert^2}

\DeclarePairedDelimiter\ip{\langle}{\rangle}

\DeclarePairedDelimiterX\setnew[2]{\{}{\}}{#1 \nonscript\;\delimsize|\nonscript\; #2} 
\newcommand*{\card}[1]{\#\left\{ #1 \right\}}
\newcommand*{\eval}{\big\vert} 
\newcommand*{\E}[1]{\Expectation\bigl\{#1\bigr\}}
\newcommand*{\lst}[2][,]{{#2}_1#1 {#2}_2#1 \dots #1 {#2}}

\DeclareMathOperator{\supp}{supp}

\DeclareMathOperator{\rank}{rank\,}

\newcommand*{\eqms}{\mathrel{\stackrel{\scriptstyle m.s.}{=}}}

\DeclareMathOperator*{\mslim}{l.i.m.}
\newcommand*{\conj}[1]{\overline{#1}}
\newcommand*{\epi}[1]{\:e^{2\pi i #1}}
\newcommand*{\empi}[1]{\:e^{-2\pi i #1}}

\newcommand*{\acorr}{autocorrelation\xspace}

\newcommand*{\nnd}{positive semi-definite\xspace} 
\newcommand*{\spd}{spd\xspace}

\renewcommand{\O}{\mathrm{\Omega}}
\newcommand{\Timedelay}{\mathrm{T}}
\newcommand{\Probspace}{\mathit{\Omega}}
\newcommand*{\Minf}{M^\infty} 
\newcommand*{\Mone}{M^1}
\newcommand*{\SetM}{U}
\newcommand*{\w}{\omega}
\renewcommand{\b}[1]{\boldsymbol{#1}}
\newcommand*{\x}{{\boldsymbol{x}}}
\newcommand*{\y}{\boldsymbol{y}}

\newcommand*{\steta}{{\boldsymbol{\eta}}}

\newcommand*{\OO}{[-\O/2,\O/2]}

\newcommand*{\f}{\boldsymbol{f}}

\newcommand{\tntn}{t,\nu;t'\!,\nu'}
\newcommand{\summ}{\sum_{\mathclap{m,m'\in\Z}}}

\newcommand*{\OPW}{\operatorname{OPW}}
\newcommand*{\StOPW}{\operatorname{StOPW}}
\newcommand*{\StL}{\operatorname{St}\!L}

\renewcommand{\SS}{S \times S}
\newcommand*{\HS}{\operatorname{HS}} 
\newcommand*{\ACP}{\operatorname{ACP}} 
\newcommand*{\Id}{\operatorname{Id}}

\newcommand*{\Trans}{\mathrm{T}}  \newcommand{\Modul}{\mathrm{M}}

\newcommand*{\sumL}[1]{\sum_{{#1}=0}^{L-1}}
\newcommand*{\sumZ}[1]{\sum\limits_{\mathclap{{#1}\in\Z}}}

\newcommand*{\RR}{\R^2}  \newcommand*{\rect}{\Box} \newcommand*{\boxfunc}{\widetilde{\chi}_{\scriptscriptstyle\rect}}
\newcommand*{\KN}{\sigma}
\newcommand*{\kernel}{\kappa}
\newcommand*{\RV}{\operatorname{RV}} \newcommand*{\Zak}{\mathcal{Z}}
\renewcommand{\H}{\boldsymbol{H}}
\newcommand*{\Zee}{\boldsymbol{\mathcal{Z}}}

\newcommand*{\vect}{\operatorname{vec}}

\newcommand*{\perm}{\sigma}
   \newcommand*{\Index}{\mathcal{J}} 
\newcommand*{\Sym}{S^{\Index}} 

\newcommand*{\tensoratom}[4]{\conj{{#1}_{#2}} \otimes {#3}_{#4} }

\newcommand*{\GG}{\tensoratom{G}{}{G}{}} 
\newcommand*{\GGL}[1]{\tensoratom{G}{#1}{G}{#1}}

\newcommand*{\isPoU}{is a permissible pattern}    

\def\Reta{R_{\scriptstyle \steta}}
\newcommand{\knkn}{k,n,k\crampedrlap{'},n\crampedrlap{'}\:} 
\newcommand{\kn}{k\crampedrlap{'},n\crampedrlap{'}\:} 
\newcommand{\bigknkn}{k,n,k'\!,n'} 
\newcommand{\knknj}{k_j,n_j,k\crampedrlap{'}_j,n\crampedrlap{'}_j\:} 
\newcommand*{\Rf}{R_{\f}}
\newcommand*{\MTc}[1][]{\Modul^{n#1} \Trans^{k#1} c } 

\makeatletter\renewcommand*\env@matrix[1][*\c@MaxMatrixCols c]{\hskip -\arraycolsep\let\@ifnextchar\new@ifnextchar\array{#1}}\makeatother

 \usepackage{tikz}
\usetikzlibrary{matrix, patterns,arrows,shadings,external}
\newcommand{\mypattern}[2] {	\pgfmathsetmacro\NN{#1 * #1};
		\draw[style=help lines] (0,0) grid (\NN,\NN);
			\foreach \m/\n in {#2}
		\filldraw[fill=mycolor, draw=black, shift={(\m-1,\n-1)}] (0,0) rectangle (1,1) node (0.5,0.5) {};  
	\draw[style=dashed, step=#1] (0.1,0.1) grid (\NN,\NN);
	\pgfmathparse{#1 - 1}
	\let\Lminusone\pgfmathresult	
	\foreach \m in {0,...,\Lminusone}
	\foreach \n in {0,...,\Lminusone}
	{	
		\pgfmathsetmacro\c{int(#1*\m+\n+1)}
		\node[below] at (\c-0.5,0) {${\eta}_{\m\n}$}; 		\node[left] at (0,\c-0.5)  { ${\eta}_{\m\n}$};  	}
		\draw[->] (-0.2,0) -- (\NN+0.3,0) node[right] {$(t,\nu)$};
	\draw[->] (0,-0.2) -- (0,\NN+0.3) node[above] {$(t',\nu') $};
}\newcommand{\mysympatternwithdots}[6] { \footnotesize
		\draw[style=help lines] (0,0) grid (#2-0.01,#2-0.01);
	\draw[style=dashed, step=#1] (0.1,0.1) grid (#2,#2);
		\foreach \m/\n in {#3}
		\filldraw[fill=mycolor, draw=black, shift={(\m-1,\n-1)}] (0,0) rectangle (1,1) node (0.5,0.5) {};  
	\foreach \m/\n in {#4}
	{	\filldraw[fill=mycolor, draw=black, shift={(\m-1,\n-1)}] (0,0) rectangle (1,1) node (0.5,0.5) {};  
		\filldraw[fill=mycolor, draw=black, shift={(\n-1,\m-1)}] (0,0) rectangle (1,1) node (0.5,0.5) {};  
	}
	\pgfmathparse{#1 - 1}
	\let\Lminusone\pgfmathresult	
	\foreach \m in {0,...,\Lminusone}
	\foreach \n in {0,...,\Lminusone}
	{	
		\pgfmathsetmacro\c{int(#1*\m+\n+1)}
		\pgfmathsetmacro\d{int(#2)<int(#1*\m+\n+1)}
		\ifthenelse{\not\equal{\d}{0.0}}{\breakforeach}{			\node[below] at (\c-0.5,0) {${\eta}_{\m\n}$}; 			\node[left] at (0,\c-0.5)  { ${\eta}_{\m\n}$};  		}
	}
	\draw[->] (-0.2,0) -- (#2+0.3,0) node[right] {$(t,\nu)$};
	\draw[->] (0,-0.2) -- (0,#2+0.3) node[above] {$(t',\nu') $};
}\newcommand{\mysympattern}[3] { \footnotesize
	\mypattern{#1}{#2};
		\foreach \m/\n in {#3}
	{	\filldraw[fill=mycolor, draw=black, shift={(\m-1,\n-1)}] (0,0) rectangle (1,1) node (0.5,0.5) {};  
		\filldraw[fill=mycolor, draw=black, shift={(\n-1,\m-1)}] (0,0) rectangle (1,1) node (0.5,0.5) {};  
	}
}\newcommand{\mynonsympattern}[3] { \footnotesize
\mypattern{#1}{#2}
\foreach \m/\n in {#3}
	\filldraw[fill=mycolor, draw=black, shift={(\m-1,\n-1)}] (0,0) rectangle (1,1) node (0.5,0.5) {};  
}
\newcommand{\myaxes}[1]
{
\pgfmathparse{int(#1)}
\let\L\pgfmathresult
\draw[->] (-0.2,0) -- (\L+0.2,0) node[below] {$(t,\nu)$};
\draw[->] (0,-0.2) -- (0,\L+0.2) node[above] {$(t',\nu')$};
}

\newcommand{\myticks}[1]
{
\foreach \x in {#1}
{   \draw[shift={(\x,0)}] (0pt,2pt) -- (0pt,-2pt);     \draw[shift={(0,\x)}] (2pt,0pt) -- (-2pt,0pt); }
}

\newcommand{\tensorsupport}
{
\myaxes{8}
\myticks{2,5,6,7}
\begin{scope}[dashed]
\foreach \x in {2,5,6,7}
{	\draw (\x,0) -- (\x,8);
	\draw (0,\x) -- (8,\x);
}
\end{scope}

\begin{scope}[ultra thick]
\foreach \a/\b in {2/5,6/7}
{
	\draw (\a,0) -- (\b,0);
	\draw (0,\a) -- (0,\b);
}
\end{scope}

\foreach \a/\b/\c/\d in {2/2/5/5,2/6/5/7,6/2/7/5,6/6/7/7}
	\filldraw[fill=mycolor, draw=black] (\a,\b) rectangle (\c,\d);

}
\newcommand{\wssussupport}
{
\myaxes{8}
\myticks{2,5,6,7}
\begin{scope}[ultra thick]
\draw (0,0) -- (7,0);
\draw (0,0) -- (0,7);
\draw[mycolor] (0.5,0.5) -- (7,7);
\end{scope}
}
\newcommand{\boxedwssussupport}
{
\myaxes{8}
\begin{scope}[dashed]
\foreach \x in {7}
{	\draw (\x,0) -- (\x,8);
	\draw (0,\x) -- (8,\x);
}
\end{scope}
\foreach \a in {0,...,6}
	\filldraw[fill=mycolor,draw=black] (\a,\a) rectangle (\a+1,\a+1);
\begin{scope}[ultra thick]
\draw (0,0) -- (7,0);
\draw (0,0) -- (0,7);
\end{scope}
}
\newcommand{\curvysupport}
{
\myaxes{8}
\myticks{1.5,7}

\tikzstyle{every node}=[]; \filldraw[fill=mycolor, draw=black]
			(1.5,1.5) .. controls +(90:3) and +(200:1) .. (3,5) node {}
						  .. controls +(20:2) and +(250:2) .. (5,7) node {}
						  .. controls +(-60:1) and +(200:0.5) .. (7,7) node {}
						  .. controls +(-110:0.5) and +(150:1) .. (7,5) node {}
						  .. controls +(-160:2) and +(70:2) .. (5,3) node {}
						  .. controls +(-110:1) and +(0:3) .. (1.5,1.5); node {}
\begin{scope}[dashed]
\foreach \x in {1.5, 7}
{	\draw (\x,0) -- (\x,8);
	\draw (0,\x) -- (8,\x);
}
\end{scope}

\begin{scope}[ultra thick]
\foreach \a/\b in {1.5/7}
{
	\draw (\a,0) -- (\b,0);
	\draw (0,\a) -- (0,\b);
}
\end{scope}
}
\newcommand{\torusdraw}[3]{	
		\pgfmathsetmacro\W{int(#3)}
	\def\A{7mm}
	\foreach \x in {0,...,\W} 
	{	\foreach \y in {0,...,\W}	
		{ 	\pgfmathsetmacro\xx{int(#1)} 
 			\pgfmathsetmacro\yy{int(#2)} 
 			\pgfmathsetmacro\opx{(\yy+0.5)/(\W+0.5)} 
 			\pgfmathsetmacro\colx{\yy/\W*150 }  			\pgfmathsetmacro\opy{(\xx+0.5)/(\W+0.5)} 
 			\pgfmathsetmacro\coly{\xx/\W*200} 
			\ifthenelse{\equal{\xx}{0.0}}			{	
					\fill[color=mypurple!\colx] (\x*\A -\A/2,\y*\A-\A/2) rectangle +(\A,\A);   
					
			}{};
			\ifthenelse{\equal{\yy}{0.0}}			{	\fill[color=myorange!\coly] (\x*\A-\A/2,\y*\A -\A/2) rectangle +(\A,\A);   					\ifthenelse{\equal{\xx}{0.0}}					{	\fill[color=mypurple!30!myorange!20] (\x*\A -\A/2,\y*\A-\A/2) rectangle +(\A,\A);   
					}{};			
			}{};
		
			\pgftext[at={\pgfpoint{\x*\A}{\y*\A}}]{\scriptsize\pgfmathprintnumber{\xx},\pgfmathprintnumber{\yy}};			
		}
	};
	\def\myboxmargin{1mm}
	\draw (-\A/2-\myboxmargin,-\A/2-\myboxmargin) rectangle (\W*\A+\A/2+\myboxmargin,\W*\A+\A/2+\myboxmargin);
	\useasboundingbox (-\A/2-\myboxmargin,-\A/2-\myboxmargin) rectangle (\W*\A+\A/2+\myboxmargin,\W*\A+\A/2+\myboxmargin);
}
\newcommand{\toruscombine}
{
	\tikzset{weak lines/.style={gray, very thin}} 
	\def\L{4}
	\def\W{3}  	\def\A{7mm}
	\tikzstyle{every to}=[bend left]
	\begin{scope}[shift={(0,0)}]
			\torusdraw{\x}{\y}{\W};
			\node[shape=coordinate] (node8) at (0,\W*\A/2) {};	
			\node[shape=coordinate] (node1) at (\W*\A,\W*\A/2) {};	
	\end{scope}	
	\begin{scope}[shift={(30mm,-30mm)}]
		\torusdraw{mod(\L-\y,\L)}{\x}{\W};
		\node[shape=coordinate] (node2) at (\W*\A/2,\W*\A) {};	
		\node[shape=coordinate] (node3) at (\W*\A/2,0) {};	
	\end{scope}
	\begin{scope}[shift={(0mm,-2*30mm)}]			
			\torusdraw{mod(\L-\x,\L)}{mod(\L-\y,\L)}{\W};
			\node[shape=coordinate] (node5) at (0,\W*\A/2) {};	
			\node[shape=coordinate] (node4) at (\W*\A,\W*\A/2) {};
	\end{scope}	
	\begin{scope}[shift={(-30mm,-30mm)}]
			\torusdraw{mod(\y,\L)}{mod(\L-\x,\L)}{\W};
			\node[shape=coordinate] (node6) at (\W*\A/2, 0) {};	
			\node[shape=coordinate] (node7) at (\W/2*\A,\W*\A) {};	
	\end{scope}
	\begin{scope}[draw,->,dashed,shorten >=5mm, shorten <=5mm]
	\draw (node7) to (node8);
	\draw (node1) to (node2);
	\draw (node3) to (node4);
	\draw (node5) to (node6);
	\end{scope}
}

\usepackage{url}
\definecolor{myorange}{RGB}{254, 196, 79}
\definecolor{mypurple}{RGB}{117, 107, 177}
\def\goodcolor{\colorlet{mycolor}{mypurple}}
\def\badcolor{\colorlet{mycolor}{myorange}}
\begin{document}
\title{\mytitle}
\author{\myauthor}
\address{ School of Engineering and Science, Jacobs University Bremen, 28759 Bremen, Germany}
\email{\{g.pfander, p.zheltov\}@jacobs-university.de}
\thanks{G.~E.~Pfander and P.~Zheltov acknowledge funding by the Germany Science Foundation (DFG) under Grant 50292 DFG PF-4, Sampling Operators.}
\thanks{\url{http://dx.doi.org/10.1109/TIT.2014.2301444}}
\thanks{\copyright\ 2014 IEEE. Personal use of this material is permitted. Permission from IEEE must be obtained for all other uses, in any current or future media, including reprinting/republishing this material for advertising or promotional purposes, creating new collective works, for resale or redistribution to servers or lists, or reuse of any copyrighted component of this work in other works.}
\date{\today}

\begin{abstract}
We develop sampling methodology aimed at determining stochastic operators that satisfy a support size restriction on the autocorrelation of the operators stochastic spreading function.  The data that we use to reconstruct the operator (or, in some cases only the autocorrelation of the spreading function) is based on the response of the unknown operator to a known, deterministic test signal.
\end{abstract}

\maketitle

\section{Introduction}\label{sec:intro} In wireless and wired communication, in radar detection, and in signal processing it is usually assumed that a signal is passed through a filter, whose parameters have to be determined from the output.
Commonly, such systems are modeled with a time-variant linear operator acting on a space of signals.
For narrow-band signals, we can model the effects of Doppler shifts and multi-path propagation as the sum of ``few'' time-frequency shifts that are applied to the sent signal. In general, the channel consists of a continuum of time-frequency scatterers: the channel is formally represented by an operator with a superposition integral
\begin{equation}\label{eq:Hf}
(H f)(x) = \iint \eta(t,\nu)  \: \Modul_\nu  \Trans_t  f(x) \d t \d\nu,
\end{equation}
where $\Trans_t$ is a \emph{time-shift} by $t$, that is, $\Trans_t \! f(x) = f(x-t)$, $t\in\R$, $\Modul_\nu$ is a \emph{frequency shift} or \emph{modulation} given by $\Modul_\nu \! f(x) = \epi{\nu x}\, f(x)$, $\nu\in \R$.
We define the Fourier transform of a function $f(x)$ to be
\[
\FT[f](\xi) = \FT_{x\to \xi} f(\xi) = \widehat{f}(\xi) = \int f(x) \empi{x \xi} \d x.
\]
It follows that
\[
\widehat{\Modul_\nu f}(\xi) = \widehat{f}(\xi-\nu) = \Trans_\nu \widehat{f(\xi)}.
\]
The function $\eta(t,\nu)$ is called the \emph{(Doppler-delay) spreading function} of $H$.

To \emph{identify} the operator means to determine the spreading function $\eta$ of $H$ from the response $H f(x)$ of the operator to a given sounding signal $f(x)$. The not necessarily rectangular support of the spreading function is known as the \emph{occupancy pattern}, and its area as \emph{spread} of the operator $H$. The fundamental restriction for the spread to be less than one has been shown to be necessary and sufficient for the identifiability of channels~\cite{Kailath, BelloMeas, Pfander, PfaWal}.
This extends results on classes of \emph{underspread} operators that are defined as those with rectangular occupancy pattern of area less than one.

This extension of the class of underspread channels to operators with spread less than one is particularly of interest in the field of sonar communication~\cite{Kilfoyle} and in the multiple input-single output channel settings.
Acoustic channels possess larger spreads than radar and wireless channels. This is due to the speed of sound being magnitudes lower than that of electromagnetic waves, resulting in time delays up to several seconds and Doppler spreads in the tens of Hertz for high-frequency channels~\cite {Baggeroer}.
Another type of channels with large values for the area of the occupancy pattern are multiple input -- single output (MIMO) channels; they combine several deterministic spreading functions into one channel, thus covering a larger region of the time-frequency plane~\cite{PfaMIMO}.

In recent work, the identifiability results~\cite{Kailath, BelloMeas, KozPfa, PfaWal} have been recast within the framework of operator sampling~\cite{PfaWalPreprint, Pfander}.
For example, in \cite{PfaWalPreprint}, concrete reconstruction formulas for deterministic operators are established, that satisfy the spread constraints mentioned above, and which resemble the Whittaker-Shannon interpolation formula
In this paper, we develop operator sampling in the stochastic setting and give analogous reconstruction formulas.

Taking into account the random nature of real-world communication environments, we model such channels with stochastic time-variant operators~\cite{Green, Kailath, BelloChar, BelloMeas}.
In this setting, the spreading function $\steta(t,\nu)$ of the operator $\H$ in \eqref{eq:Hf} is a random process\footnote{
Here, and in remainder of the paper, random functions and operators are denoted by boldface characters, $\conj{x}$ is a complex conjugate of $x$, $x^*$ conjugate transpose. $\EXP$ stands for expectation, and $\mu$ for both 2D area and 4D volume.
}
 indexed by $(t,\nu)$ that is to be recovered from the output process $\H f(x)$ indexed by $x$.
For the purposes of this paper, it would be enough to think of $\steta(t,\nu; \w)$ as a $(t,\nu)$-indexed family of random variables on a common probability space $\Probspace$, or as a random process with instances from $L^2(\RR)$.

In the sibling paper \cite{PfaZh03} we develop and use the theory of stochastic modulation spaces to rigorously define and prove identification results for operators with spreading functions belonging to a class of \emph{generalized} random processes --- including delta functions and white noise.
The norm inequalities that are proven there are essential to justify use of delta trains as sounding signals.
However, in this paper, we are going to ignore these subtleties and treat distributions on par with Lebesgue integrable functions, for ease of exposition only.
The fine points of the underlying mathematical analysis will be mentioned in a few side remarks.

\subsection{Operator sampling theory in the historical perspective}
The progress of operator identification theory largely follows the evolution of the related theory of function sampling. The two major directions of generalization are the introduction of stochasticity (stationary and non-stationary) and the removal of the requirement that the \enquote{bandlimitation} is rectangular. For convenience, we summarize this development in \autoref{table:development}.

\newcommand{\vtt}[2]{\multirow{#1}{*}{\begin{sideways} #2 \end{sideways}}}
\begin{table*}
\begin{adjustwidth}{-1cm}{-1cm}
\normalsize
 \renewcommand\arraystretch{1.3}
\begin{center}
\begin{tabular}{lcccc}
\toprule
					  &                    &                                & rectangular             & non-rectangular                      \\  \midrule
\vtt{6}{Sampling of } & \vtt{3}{functions} & Deterministic                  & \citeauthor{Shannon}      & \citeauthor{Klu}                       \\
					  &                    & Stochastic stationary          & \citeauthor{Lloyd} 		  & \citeauthor{Lloyd}                     \\
					  &                    & Stochastic                     & \citeauthor{Lee}          & \citeauthor{Lee}                       \\
\cmidrule[\heavyrulewidth]{2-5}
                      & \vtt{4}{operators} & \multirow{2}{*}{Deterministic} & \citeauthor{Kailath},      & \citeauthor{Pfander},                   \\
                      &                    &                                & \citeauthor{KozPfa}       & \citeauthor{PfaWal}                    \\
\cmidrule{4-5}
					  &                    & Stochastic stationary          & \citeauthor*{OPZ}          & \autoref{thm:wssus}                     \\
					  &                    & Stochastic                     & \autoref{thm:stoch_tensor} & \autoref{thm:stoch_nontensor_rectified} \\
\bottomrule
\end{tabular}
\caption{Development of function and operator sampling.
\label{table:development}}
\end{center}
\end{adjustwidth}
\end{table*}

For example, based on the classical Shannon-Nyquist sampling theorem, a corresponding result for bandlimited stationary stochastic processes was proven in great generality by Lloyd \cite{Lloyd}.
We cite it here in the form given in a classic book by Papoulis.
\begin{theorem}\cite[p.~378]{Papoulis}
If a stationary process $\x(t)$ is bandlimited, that is, if its  power spectral density $S(\xi) \coloneqq \FT_{\tau\to\xi} \EXP\{\x(t)\conj{\x(t+\tau)}\}$ is integrable and $\supp S(\xi) \subset \OO$, then we can recover $\x(t)$ in the mean-square sense from the samples taken at rate $\O = \Timedelay^{-1}$. In fact,
\[
\x(t) = \mslim_{N\to \infty} \sum_{n=-N}^N \x(n \Timedelay) \, \frac{\sin \pi\O(t-n \Timedelay)}{\pi\O(t-n \Timedelay)}.
\]
\end{theorem}
The requirements of stationarity and the bandlimitation of the spectrum to the symmetric interval were later relaxed by \citeauthor{Klu} and \citeauthor{Lloyd}. 
In his 1963 groundbreaking paper~\cite{Kailath}, Kailath realized that for a deterministic time-variant channel to be identifiable, it is necessary and sufficient that the product $\O \Timedelay$ of the maximum time delay $\Timedelay$ and maximum Doppler spread $\O$ is not greater than one.
Since, channels were called \emph{underspread} whenever $\O \Timedelay<1$, and \emph{overspread} if $\O \Timedelay>1$~\cite{VanTrees}.
The insight of Kailath has been generalized and formalized by \citeauthor{KozPfa}.

Following in Kailath footsteps, the seminal paper of \citeauthor{BelloChar} lays the groundwork for channel sampling and characterization tools and vocabulary.
In the sequel \cite{BelloMeas} Bello further argues that it is not the product $\O \Timedelay$ that matters for identification of a deterministic time-variant channel, but rather the \emph{spread}, or the area of what he calls an \emph{occupancy pattern}, that is, the not necessarily rectangular support of the spreading function $\supp \eta$.
In particular, Bello's assertion has been put into a rigorous mathematical framework and was proven using novel tools from Gabor analysis by \citeauthor{PfaWal}.
Also, the ideas developed in \cite{BelloChar} were used to estimate the capacity of the channels with non-rectangular spread~\cite{DMBSS}.

A brief comment of Kailath~\cite{Kailath} suggests sufficiency of $\O \Timedelay\leq 1$ for the identification of stationary stochastic channels as well as deterministic, and Bello treats this question as a side matter, more interested in developing the estimator for $\steta(t,\nu)$ when the output has been contaminated by additive noise.

As with the development of function sampling, a simpler stationary model for operator identification has seen most research.
The channel has the property of \emph{wide-sense stationarity with uncorrelated scattering} (WSSUS), that is, the autocorrelation function of $\steta(t,\nu)$ has the form
\begin{align*}\label{eq:Reta}
\Reta(\tntn) &= \E{\steta(t,\nu)\, \conj{\steta(t',\nu')}} \\
&= \delta(t-t')\,\delta(\nu-\nu')\,C_\steta(t,\nu).
\end{align*}
In other words, taps at different delays are uncorrelated and stationary. The function $C_\steta(t,\nu)$ is known as the \emph{scattering function} of $\H$. It completely characterizes the second-order statistics of $\steta$ and represents the power spectral density of the transfer function of the channel. This means that the scattering function represents the expected behavior of the operator.
Two common types of methods to identify the scattering function are deconvolution and direct measurement methods~\cite{artes2004unbiased, Gaarder, NguyenSenadji,KayDoyle}.
In \cite{PfaZh04} we apply the methodology developed here and in \cite{PfaZh03} to study WSSUS channels in depth.
\autoref{thm:diagonal} below guarantees identifiability of a WSSUS channel whenever its scattering function is merely compactly supported.

In this paper, we address a more general problem of stochastic spreading function reconstruction and stochastic operator sampling and identification for not necessarily WSSUS channels.

\subsection{Overview of the paper}
Deterministic identification results in \cite{KozPfa, PfaWal} allow for the recovery of the deterministic spreading function whenever its support has area less than one.
We discuss this in detail in \autoref{ssec:samp_det_op}.
It is easy to see that in the case of a stochastic spreading function, such deterministic reconstruction formulas are still applicable to each random variate, allowing for the recovery of an \emph{instance spreading function} $\steta(t,\nu; \w)$ from the response $\H(\w) f(x)$, where $\w$ is an element of the sample space $\Probspace$.
However, on its own, each instance will provide little information about the average behavior of the operator $\H$.
Secondly, it is possible that the 4-dimensional volume of the support of the 4D autocorrelation function $\Reta(\tntn)$ is less than one, while some instances $\steta(t,\nu; \w)$ have 2D area greater than one.

The contributions of this paper follow. In \autoref{ssec:samp_stoch_op_rect} we consider the case when the 4-dimensional autocorrelation function $\Reta(\tntn)$ of the operators' spreading function is supported on a 4D region $\SetM = \supp \Reta(\tntn)$ that can be expressed as a tensor product of some 2D region $S$ with itself.
In this scenario, we prove that it is possible --- and give an explicit reconstruction formula \eqref{eq:reconstruction} --- to recover the stochastic spreading function of the channel in the mean-square sense  from the response $\H f$ of a channel to a periodic weighted delta train, provided that the set $S$ occupies a region of area less than one.
This case will include the special case deterministic operators as their \acorr functions satisfy $\Reta(\tntn) = \eta(t,\nu)\, \conj{\eta(t',\nu')}$.

In \autoref{ssec:samp_stoch_op_nonrect} the case of an arbitrary support region $\SetM = \supp \Reta(\tntn)$ is considered.
With some abuse of nomenclature, we will also say that we can \emph{(stochastically) identify} a stochastic operator if we can recover the 4-dimensional autocorrelation function $\Reta(\tntn)$ (and not the stochastic spreading function, as in the previous case) from the autocorrelation function of the response $\H f$ of the channel to a periodically weighted delta train (and not the stochastic response itself).
In \autoref{thm:stoch_nontensor_rectified} we prove that in this sense we can identify the operator $\H$ provided that the support pattern of the autocorrelation function is \emph{permissible}.

Analogous to the deterministic criterion $\mu(S)<1$, the requirement $\mu(\SetM) < 1$ is also necessary for stochastic identifiability of a stochastic operator, as we show in an sibling paper~\cite{PfaZh03}.
In this paper, we demonstrate patterns that correspond to regions of 4D volume less than one but are nonetheless \emph{unidentifiable} by our methods.

Permissibility of a pattern is a geometric property linking the operator sampling theory to finite Gabor frame theory.
After giving preliminary remarks on Gabor frames in finite dimensions in \autoref{ssec:Gabor} below,
in \autoref{sec:frame_theory} we show that the support patterns of the the \acorr functions of the operators are in one-to-one correspondence with the column subsets of special Gabor frames.
We provide a (partial) classification of the autocorrelation patterns and analyze several phenomena that emerge for Gabor frames of higher dimension \autoref{ssec:L=2} and \autoref{ssec:L=2_and_up}.

\subsection{Finite-dimensional Gabor frames}\label{ssec:Gabor}
A deterministic operator is a particular case of the stochastic operator with a degenerate probability distribution, so the results for stochastic operator identification must necessarily be compatible with the deterministic ones.
In the theory of deterministic operator identification, the existing operator identification proofs pivot on the so-called \emph{Haar property} of Gabor frames.
A finite-dimensional Gabor frame in $\C^L$ is defined as
\begin{equation}\label{eq:Gabor.frame}
G \coloneqq \{ \MTc \}_{k,n=0}^{L-1},
\end{equation}
where the finite-dimensional translation operators $\Trans^k$ and modulation operators $\Modul^n$ operating on a vector $c\in \C^L$ are given by
\begin{equation}\label{eq:define.fin.dim.TM}
\Trans^k c[p] = c[p-k] \quad \text{ and } \quad \Modul^n c[p] \coloneqq \epi{n p / L} \: c[p].
\end{equation}
A frame has the \emph{Haar property} whenever its elements are in general linear position, that is, any subset of $L$ elements  is linearly independent. A finite-dimensional Gabor frame $G$ generated by window $c$ as defined in \eqref{eq:Gabor.frame}, has the Haar property for almost every choice of window $c$, given that the ambient dimension $L$ is prime \cite{LPW}.

Application of our methods to the stochastic case spawns a more peculiar Gabor frame on $\Z_L \times \Z_L$, the properties of which differ in the key respect of linear independence of its subsets. In particular, such a frame has small subsets that are linearly dependent for any choice of window $c$, even when the parameter $L$ is prime.

 \section{Operator sampling}\label{sec:samp_op}
	\subsection{Deterministic operators}\label{ssec:samp_det_op}Equivalently to \eqref{eq:Hf}, any operator $H$ acting on one-variable signals can be represented with its symbols, 
\begin{enumerate}
\item its \emph{time-varying impulse response} $h(x,t) = \int \eta(t,\nu)\empi{\nu (x-t)} \d \nu $, then 
\[(Hf)(x) = \int h(x,t)f(x-t) \d t,\]
\item its \emph{kernel} $\kernel(x,t) = h(x,x-t) = \int \eta(x-t,\nu) \empi{\nu t} \d \nu,$ then  
\[(Hf)(x) = \int \kernel(x,t)f(t) \d t,\]
\item and its \emph{Kohn-Nirenberg symbol} $\KN(x,\xi) = \FT_{t \to \xi} h(x,t)$, then 
\[ (Hf)(x) = \int \KN(x,\xi) \widehat{f}(\xi) \epi{x \xi} \d\xi. \]
\end{enumerate}
All these symbols of $H$ can be transformed into each other using partial Fourier transforms and area-preserving shears, such as $\mathcal{I}_2f(x,t) \coloneqq f(x+t,x)$~\cite{Gro}. In particular, the spreading function and the Kohn-Nirenberg symbol are related through
\[ 
\KN(x,\xi) \coloneqq \FT_s\eta(t,\nu),
\]
where the \emph{symplectic Fourier transform} is given by 
\[ 
\left(\FT_s \eta\right)(x,\xi) = \iint \eta(t,\nu)\empi{(\nu x - \xi t)} \d t \d\nu.
\]

In the following, it will sometimes be advantageous to present the results with $h$, $\kernel$, or $\KN$ instead of $\eta$ and we will not hesitate to do so. However, we formulate our results primarily with $\eta$ due to its particularly simple relationship to the short-time Fourier transform. 
On the Schwartz dual space $\S'$ of tempered distributions, we define the \emph{short-time Fourier transform} to be 
\[ 
V_{\phi} f(t,\nu) \coloneqq \ip{f, \Modul_\nu \Trans_t \,\phi}
\]
for any $\phi$ in the Schwartz space $\S$.
Then for all $f,\phi\in\S$ we have the useful equality 
\begin{equation*}\label{eq:eta.vs.V}
\ip{Hf,\phi} = \ip{\eta, V_f \phi}. 
\end{equation*}
The inner product $\ip{\cdot,\cdot}$ is taken to be conjugate linear in the second variable. 

We will say that $H$ belongs to an \emph{operator Paley-Wiener space} if the support of the spreading function $\eta(t,\nu) =  \FT_{s} \KN(x,\xi)$ is contained in a compact subset of the time-frequency plane,
\[
\OPW(S) \coloneqq  \Big\{H \colon L^2(\R) \to L^2(\R),  \text{such that } \KN\in L^2(\R^2) \text{ with } \supp \FT_s\KN \subseteq S \Big\}.
\]
Colloquially, we refer to $H$ as \emph{``bandlimited''} to $S$. 

The connection of the operator sampling theory to the more established function sampling is best observed on the following theorem for Hilbert-Schmidt operators with the spreading function supported on a rectangle in the time-frequency plane. 
An operator $H \colon L^2(\R^n)\to L^2(\R^n)$ is \emph{Hilbert-Schmidt} ($H \in \HS(\R^n)$) whenever its spreading function $\eta(t,\nu)\in L^2(\R^{2n})$. 
\begin{theorem}\cite[Theorem 1.2]{Pfander}
\label{thm:samp_HS}
For $H \in \HS(\R)$ such that 
$\supp \eta(t,\nu) \subseteq[0,\Timedelay)\times \OO$
and $\O\Timedelay\leq 1$ we have 
\begin{equation*}
\norm{H\sum_{k\in \Z} \delta_{k \Timedelay}}_{L^2(\R)} = \Timedelay \norm{\eta}_{L^2(\R^2)},
\end{equation*}  and $H$ can be reconstructed by 
\[ \kernel(x+t,x) = \sum_{n\in\Z}(H\Shah)(t+n \Timedelay)\frac{\sin \pi \Timedelay(x-n)}{\pi \Timedelay(x-n)},\]
with convergence in $L^2(\R^2)$. 
Here and in the following, we denote $\delta_{a}(t) \coloneqq \delta(t - a)$. 
\end{theorem}
Note that if $H$ is a multiplication operator with a bandlimited multiplier, \autoref{thm:samp_HS} reduces to the classic Shannon sampling theorem. 

\begin{remark}
In order to accommodate delta functions and other generalized functions as identifiers we require that $\steta(t,\nu)$ is a mapping from the dual $\Minf(\R^2)$ of the Feichtinger algebra $\Mone(\R^2)$ into the Hilbert space of zero-mean random variables $\RV(\Probspace)$.
Here, the \emph{modulation space} $\Mone(\R^n)$ is defined via the finiteness of the norm $\norm{f}_{M^1(\R^n)} \coloneqq \norm{ V_g f}_{L^1(\R^{2n})}<\infty$, where $L^1(\R^n)$ is the space of Lebesgue integrable functions, and $\Minf(\R^n) \coloneqq (\Mone(\R^n))'$ is its continuous dual \cite{Gro}. 

Functional analytic arguments of \cite{KozPfa, PfaWalFeich, Pfander, PfaZh03} show that whenever $\eta(t,\nu)$ is compactly supported (which is a physically reasonable assumption taken here) it is possible to extend the domain of $H$ from $L^2$ to the whole of $\Minf$ and further to the Wiener amalgam space $W(\FT L^\infty(\R^{2d}),\ell^1)$. 
Therefore, $\eta(t,\nu)$ is well-defined as a linear functional on the corresponding predual space $W(\FT L^1(\R^{2d}), \ell^\infty)$ that includes certain tempered distributions, in particular, weighted delta trains $\Shah_c =\sum_{k \in \Z} c_k \: \delta_{k\Timedelay}$ as input for $H$. 
\end{remark}

Below, \autoref{thm:det} gives a deep generalization of \autoref{thm:samp_HS}. 
It provides guarantees for recovery of those operators whose spreading function is supported on a set of arbitrary shape --- not just a rectangle --- as long as the area of the support is less than one. Additionally, it shows that it is possible to replace $\sin x / x$ with smooth functions with better decay through the use of partitions of unity generated by continuous functions. For $a, \eps>0$ we say $r(t)$ generates an \emph{$(a,\eps)$-partition of unity} $\{ r(t+ak) \}_{k\in \Z}$ whenever 
\[ 
r(t) = 0 \text{ if } t\not\in (-\eps, a+\eps)\quad \text{ and } \quad \sum_{k\in\Z} r(t+ak) = 1.
\]

\begin{definition}
\label{defn:S.is.rectified}
The set $S$ is \emph{$(a,b,\Lambda)$-rectified} if it can be covered by $L=\frac{1}{ab}$ translations $\Lambda \coloneqq \{(k_j, n_j)\}_{j=0}^{L-1}$ of the rectangle $\rect \coloneqq [0,a)\times[0,b)$ along the lattice $a\Z\times b\Z$: 
\[ 
S \subset \bigcup_{j=0}^{L-1} \rect+(a k_j, b n_j).
\]
\end{definition}

\begin{lemma}
\label{lem:cover_det}
For any compact set $S$ of measure $\mu(S)<1$ contained within $[0,\Timedelay)\times \OO$ there exist $a,b>0$ and $\Lambda$ such that $S$ is $(a,b,\Lambda)$-rectified, $L = \frac{1}{ab}$ is prime, $a<\frac{1}{\Timedelay}$ and $b<\frac{1}{\O}$. 
\end{lemma}
\begin{proof}
This is a standard result from a theory of Jordan domains. A proof can be found in \cite{Folland}. Any compact set with Jordan content less than one can be covered with $L$ rectangles of cumulative area $L\cdot \frac{1}{L} = 1$, and it is easy to see that we do not lose generality requiring $a,b$ to be small enough to satisfy $a<1/\Timedelay$, $b<1/\O$ and $L$ prime. 
\end{proof}

\begin{theorem}\cite[Theorem 3.1]{PfaWal}, \cite{PfaWalPreprint} \label{thm:det}
Let $S\subset [0,\Timedelay)\times \OO\subset \RR$ be a compact set with measure $\mu(S)< 1$, such that $S$ is $(a,b,\Lambda)$-rectified as in \autoref{defn:S.is.rectified} (with $\O \Timedelay$ not necessarily smaller than 1). Then there a vector $c\in \C^L$, and the test signal 
\[\Shah_c=\sum_n c_{n \bmod L} \:\delta_{n/L},\]
such that for any $H\in \OPW(S)$ 
\begin{align*} h(x+t,t) &= aL\sumL{j} \sumZ{q} a_{j,q}(H\Shah_c)(x-a(k_j+q))  \\
	&\quad \times r(x-a k_j) \:  \phi(t+a(k_j+q)) \epi{b n_j t}
\end{align*}
with convergence in $L^2(\R^2)$. Here, the coefficients $a_{j,k}$ are uniquely determined by the choice of $\{c_n\}$, and $r(t)$, $\phi(t)$ are any functions such that $r(t-ak)$ and $\widehat{\phi}(\gamma-bn)$ are $(a,\eps)$- (respectively, $(b,\eps)$-) partitions of unity in time (respectively, frequency) domains, with $\eps>0$ dependent on $S$. 
\end{theorem}

Since we can always rectify a compact $S$ with measure $\mu(S)<1$, \autoref{thm:det} holds for a wider class of regions, namely, of all regions $S$ whose so-called Jordan outer content is less than one \cite{PfaWal}.  	\subsection{Stochastic operator Paley-Wiener spaces}\label{ssec:samp_stoch_opw}Let $\H$ be a stochastic operator with integral representation
\begin{equation*}\label{eq:H}
\H f=\iint \steta(t,\nu) \: \Modul_\nu  \Trans_t  f \d t \d\nu,
\end{equation*}
and stochastic {\it spreading function} $\steta(t,\nu)$ a zero-mean random process such that $\steta(t,\nu; \w) \in L^2(\RR)$ for all $\w\in \Probspace$.
We denote the space of all such stochastic processes by $\StL^2(\R^2)$.
The \emph{autocorrelation} of the spreading function is given by
\[
\Reta(\tntn) \coloneqq \E{\steta(t,\nu) \, \conj{\steta(t',\nu')}}.
\]
\begin{definition}
We say that $\H$  is \emph{a stochastic Paley-Wiener operator bandlimited} to $\SetM$, whenever the support of $\Reta(\tntn)$ is contained within a closed set $\SetM \subseteq \R^4$, that is,
\[
\StOPW(\SetM) = \Big\{\H \colon L^2(\R) \to \StL^2(\R) \colon \supp \Reta(\tntn) \subseteq \SetM \text{ with $\SetM$ closed}\Big\}.
\]
\end{definition}
In this section, we always assume $\SetM$ already rectified, as done in the deterministic case.

The symmetries of the autocorrelation function require a symmetric rectification, defined below.
However, it will not cause any confusion that whenever we speak of a rectified 4D region, we always mean a symmetrically rectified one.
\begin{definition}\label{defn:M.is.rectified}
The set $\SetM$ is \emph{$(a,b,\Lambda)$-symmetrically rectified} if it can be covered by the translations
$\Lambda \coloneqq \left\{ \lambda_j \coloneqq (k_j, n_j, k'_j, n'_j) \right\}_{j=0}^{L^2-1}$
of the prototype parallelepiped $\rect^2 \coloneqq [0,a)\times[0,b) \times [0,a) \times [0,b)$ along the lattice $(a\Z \times b\Z)^2$
\begin{equation*}\label{eq:M.is.rectified}
\SetM \subset \bigcup_{j=0}^{L^2-1} \rect^2+ (a k_j, b n_j, a k_j', b n_j')
\end{equation*}
such that the 4D volume of $\rect^2$ is small: $ab=\frac{1}{L}$, $L$ is prime, and the right-hand side is a symmetric set.
\end{definition}

It is easy to see that the above requirements imply that the volume of $\SetM$ satisfies $\mu(\SetM) \leq 1$, and conversely, it can be shown with little work that any symmetric compact set with Jordan content less than one can be rectified with a sufficiently large $L$~\cite{Folland}.
This restriction on the area is crucial for operator identification.
Before giving results for general operators whose \acorr functions are supported on arbitrary sets $\SetM \subset \R^4$, we look at the sets of a special kind, where $\SetM=S\times S$.
A general case will be considered in \autoref{ssec:samp_stoch_op_nonrect} below.

  	\subsection{Sampling of stochastic operators supported on \texorpdfstring{$S\times S$}{SxS}}\label{ssec:samp_stoch_op_rect}Under the special circumstance that $\SetM$ can be represented as a product $S\times S$ of some set $S$ in $\RR$, or rather, if we assume $\SetM$ to be $(a,b,\Lambda)$-symmetrically rectified, $\Lambda = \Gamma \times \Gamma$ for some $\Gamma \subset \Z\times \Z$ such that $\abs{\Gamma} = L$, the operator (via its spreading function) can be reconstructed directly from the output of the weighted delta train input with an explicit formula similar to the one presented by Pfander and Walnut~\cite{PfaWalSampta}.

We define the \emph{non-normalized Zak transform}  $\Zak \colon L^2(\R) \to L^2\left([0, aL)\times[0,b)\right)$ as
\begin{equation*}\label{eq:defn.Zak}
\Zak f(x,\nu) \coloneqq \sumZ{n} f(x-a n L)\epi{a L n \nu}.
\end{equation*}
We say that a series of random processes $\sum_{n\in \Z} \x_n(t)$ converges to $\x(t)$ in \emph{mean-square} sense
\[
\x(t) \eqms \sum_{n\in \Z} \x_n(t) \quad \Leftrightarrow \quad \lim_{n\to \infty} \E{\absq{ \x(t) - \sum_{n=-N}^N \x_n(t)}} = 0.
\]

\begin{theorem} \label{thm:stoch_tensor}
Let $\H\in\StOPW(\SS)$ such that the compact set $S\subset \RR$ has measure $\mu(S)<1$. Then there exist $L$ prime, $a,b>0$ with $ab=\frac{1}{L}$, a complex vector $c\in \C^L$ and a sequence $\{ \alpha_{jp} \}_{j,p=1}^L$ depending only on $c$ such that we can reconstruct $\H$ from its response to an $L$-periodic $c$-weighted delta train $\Shah_c = \sumZ{k} c_{k \pmod L} \: \delta_{a k}$
via
\begin{align*}\label{eq:rec}
\steta(t,\nu) &\eqms aL \sumL{j}\sum_{p=0}^{L-1}\alpha_{j p} \empi{a(\nu-b n_j)(p+ k_j)}\boxfunc(t,\nu) \\
&\quad \times  \left(\Zak_{a L}\H \Shah_c\right)(t+a(p-k_j),\nu-b n_j),
\end{align*}
where the translations $\boxfunc(t+a k,\nu+bn)$ generate an $(a,b,\eps)$-partition of unity in the time-frequency domain.
\end{theorem}
This means that it is always possible to find $c$ such that any two operators from $\StOPW(S\times S)$ can be distinguished from their response to the pilot signal $\Shah_c = \sumZ{k} c_{k \pmod L}\: \delta_{ak}$.

\begin{proof}
Let $a,b>0$, $L$ -- prime with $ab = \frac{1}{L}$ be fixed numbers to be chosen later. This choice will depend on the support of $\steta(t,\nu)$. Let $c\in \C^{L}$; indices of $c$ should always be understood modulo $L$.
\allowdisplaybreaks
\begin{align*}
\H\Shah_c(x) &= \iint \steta(t,\gamma)\epi{\gamma x}\:\Shah_c(x-t)\d t\d \gamma  \\
&= \iint \steta(t,\gamma)\epi{\gamma x}\sumZ{k} c_k \delta(x-t-ak)\d t\d \gamma \displaybreak[3]\\
&= \int \sumZ{k} c_k \, \steta(x-ak,\gamma)\epi{\gamma x} \d \gamma \displaybreak[3]\\
&= \sum_{k\in \Z} c_{k+p}\int \steta(x-a(k+p),\gamma) \epi{\gamma x} \d \gamma,
\end{align*}
with all the equalities holding $\w$-surely, and the last equality being true for arbitrary $p\in \Z$.
Applying Zak transform to both sides, we get
\newcommand*{\ministeta}[1]{\steta({\scriptstyle #1})}
\begin{align*}
\allowdisplaybreaks
\textstyle
\Zak (\H \Shah_c) (x,\nu)
&= \sum_{n\in\Z}\H\Shah_c(x-anL)\epi{ anL\nu}  \\
& = {\sum\limits_{\mathclap{k,n\in\Z}} c_{k+p} \int \ministeta{x-a(k+p+nL),\gamma}\!\epi{(a  n L  \nu + \gamma (x- a n L))}\!\d \gamma }\\
\intertext{ substitute $k=k+nL$ }
 & = \sumZ{k} c_{k+p}  \int \ministeta{ x-a(k+p),\gamma} \epi{\gamma x} \! \sumZ{n} \epi{a n L(\nu-\gamma)}\d \gamma \\
 \intertext{ by Poisson summation formula $\sum \epi{n q x} = q^{-1}\sum \delta(q x-n)$ with $q = aL = \nicefrac{1}{b}$ }
&=  \sumZ{k} c_{k+p}  \int \ministeta{x-a(k+p),\gamma} \! \epi{\gamma x} \frac{1}{aL} \sum_{n\in\Z}\delta({\scriptstyle aL(\nu-\gamma)-n}) \d \gamma  \\
&=  \sumZ{k} c_{k+p}  \int \ministeta{x-a(k+p),\gamma}\epi{\gamma x} \: b\sum_{n\in\Z}\delta(\nu-\gamma-bn) \d \gamma \\
\intertext{carry out integration in $\gamma$ and set $n=-n, k=-k$}
&= b \sumZ{k,n}  c_{p-k}  \steta(x-a(p-k), \nu+bn)\epi{(\nu+b n) x }.
\end{align*}
Substitute $t=x-ap$ and observe in the exponent $(\nu+b n)x= \nu (t + a p) + b n t  + \frac{n}{L}p$. For brevity, we define $\Zee_p(t,\nu)$ as
\begin{equation}\label{eq:Z=G_tilta}
\begin{split}
\Zee_p(t,\nu) & \coloneqq b^{-1} \empi{\nu (t+ap)} \:\Zak \H \Shah_c(t+ap,\nu) \\
& = \sumZ{k,n} c_{p-k} \epi{p n/L} \:\underbrace{\steta(t+ak ,\nu+bn)\epi{b n t}}_{\steta_{k,n}(t,\nu)}.
\end{split}
\end{equation}
For all $t,\nu\in \rect$ and all $\w\in\Probspace$ we then have a mixing matrix equation
\begin{equation}\label{eq:mixing_full}
\Zee(t,\nu) = \sum_{k,n\in\Z} \MTc \: \steta_{k,n}(t,\nu),
\end{equation}
where $\Zee(t,\nu) = [ \Zee_p(t,\nu) ]_{p=0}^{L-1}$ is a vector-valued function on $(t,\nu)\in \rect$.

By \autoref{lem:cover_det}, the set $S$ can be $(a,b,\Gamma)$-rectified, that is, there exists a collection of indices
$\Gamma = \{ (k_j, n_j) \}_{j=0}^{L-1}$ such that $\steta_{k,n}(t,\nu) \equiv 0$ on $(t,\nu)\in \rect$ for any $(k,n)\not\in \Gamma$.
Therefore, by \autoref{lem:strip} below, the infinite sums in \eqref{eq:mixing_full} can be trimmed to finite sums to obtain
\begin{equation}\label{eq:mixing}\begin{split}
\Zee(t,\nu) &= \sum_{k,n\in\Z} \MTc \: \steta_{k,n}(t,\nu) \\
&\eqms \sum_{(k,n)\in\Gamma} \MTc \: \steta_{k,n}(t,\nu) \\
&= G\eval_\Gamma \, \steta_\Gamma(t,\nu),
\end{split}\end{equation}
where $G$ is the $L\times L^2$ Gabor matrix of all time-frequency shifts of $c\in\C^L$ given by \eqref{eq:Gabor.frame}, $G\eval_\Gamma$ is the submatrix of $G$ corresponding to columns indexed by $\Gamma \bmod L$, and
\[ \steta_\Gamma(t,\nu) = \Big[ \steta_{k_j, n_j}(t,\nu)\Big]_{(k_j, n_j) \in \Gamma}
\]
is a column vector of nonzero patches of $\steta(t,\nu)$ defined by \eqref{eq:Z=G_tilta}.

Without loss of generality, no two indices in $\Gamma$ correspond to the same column of $G$, that is, $(k, n)\neq (k', n') \bmod L \in \Gamma$ for all $(k,n)\neq (k', n')\in \Gamma$. Such collisions can always be avoided by choosing a different rectification $(a',b',\Gamma')$ such that the entire support set $S$ is within $[0,a'L') \times [0, b'L')$, which can always be achieved using \autoref{lem:cover_det}.
By assumption, $\abs{\Gamma}\leq L$, therefore, $G\eval_\Gamma$ is invertible for some complex vector $c\in\C^L$.\footnote{
In fact, for prime $L$ the set of $c\in\C^L$ such that \emph{every} $L\times L$ submatrix of $G$ is invertible is a dense open subset of $\C^L$~\cite{LPW}.
Furthermore, all entries of $c$ can be additionally chosen to have absolute value one.
}
Denote
\[
A = \left(G\eval_\Gamma\right)^{-1} = [\alpha_{jp}]_{j,p=0}^{L-1}
\]
the inverse of $G\eval_\Gamma$. The solution to \eqref{eq:mixing} can now be found, giving
\begin{align*}
 \steta_{k_j,n_j}(t,\nu) & \eqms \left(G\eval_\Gamma\right)^{-1} \Zee(t,\nu) \\
& = aL \empi{\nu  t} \sum_{p=0}^{L-1} \alpha_{jp} \empi{\nu ap} \Zak  \H \Shah(t+ap, \nu).
\end{align*}
We can now combine the patches $\steta_{k,n}(t,\nu)$ into the whole spreading  function $\steta(t,\nu)$. In the mean-square sense we have
\begin{align*}
 \steta(t,\nu) &\eqms \sumL{j} \steta(t,\nu)\boxfunc(t - a k_j ,\nu - b n_j)  \\
&\eqms \sumL{j} \steta_{k_j,n_j}(t - a k_j,\nu - b n_j) \empi{ b n_j (t- a k_j)}  \boxfunc(t - a k_j, \nu - b n_j) \\
&\eqms aL \sumL{j}\sum_{p=0}^{L-1}\alpha_{jp} \empi{(\nu(t + a(p-k_j)) - bn_j a p)} \boxfunc(t,\nu)\\
& \qquad \times \Zak \H \Shah_c (t+a(p-k_j),\nu-bn_j). \qedhere \end{align*}
\end{proof}

\begin{lemma}
\label{lem:strip}
Let $\steta(t,\nu)$ be such that $\supp \Reta (\tntn) \subset S\times S$ with $S\subset \R^2$ that is $(a,b,\Gamma)$-rectified. Then for all $(t,\nu)\in [0,a)\times [0,b)$,
\begin{align*}
\sumL{k}\sum_{l,n \in \Z} a_{k,n,l}(t,\nu) \: \steta(t-a(lL+k),\nu+bn) \eqms \sum_{\gamma\in\Gamma} a_{\gamma}(t,\nu) \: \steta_\gamma(t,\nu),
\end{align*}

where $\steta_\gamma(t,\nu) \coloneqq \steta(t-a(lL+k),\nu+bn)$.
\end{lemma}
\begin{proof}
Denote $Q=\Z/L\Z \times L\Z\times \Z$ the set of all indices on the left hand side. Consider
\begin{equation*}\label{eq:meansquare}
\EXP\Bigl\{\absq{ \sum_{\gamma \in Q} a_\gamma \steta_\gamma -  \sum_{\gamma\in\Gamma} a_\gamma \steta_\gamma}\Bigr\}
= \sum_{\mathclap{\gamma_1,\gamma_2\in Q\setminus\Gamma}} a_{\gamma_1}\conj{a_{\gamma_2}}\: \E{\steta_{\gamma_1}\conj{\steta_{\gamma_2}}}.
\end{equation*}
The rectification of $\Reta$ guarantees that
\[
\E{\steta_{\gamma_1} (t_1,\nu_1)\,\conj{\steta_{\gamma_2}(t_2,\nu_2)}} = 0
\]
 whenever either $\gamma_1$ or $\gamma_2$ are not in $\Gamma$. We conclude
\[ \sum_{\gamma \in Q} a_\gamma \steta_\gamma(t,\nu) \eqms \sum_{\gamma\in\Gamma} a_\gamma \steta_\gamma(t,\nu). \qedhere\]
\end{proof}
\begin{remark}\label{rem:WSSUS1}
Since the autocorrelation functions of WSSUS operators have the special form
\[ \Reta(\tntn) = \delta(t-t')\,\delta(\nu-\nu')\,C_\steta(t,\nu),\]
\autoref{thm:stoch_tensor} is applicable to WSSUS operators whenever the area of the support of the scattering function (also known as the \emph{spread} of the operator)  satisfies $\mu(\supp C_\steta(t,\nu)) < 1$.
However, it is intuitively clear that it is excessive to cover a 4D diagonal (a set of measure zero in $\R^4$) with a single 4D parallelepiped of volume  one. Indeed, covering it with a string of small parallelepipeds (as in \autoref{fig:wssus.supports}) is enough for identification and allows recovery of $C_\steta$ for $\supp C_\steta$ bounded, of arbitrary size. This comes as a corollary from \autoref{thm:wssus} and \autoref{thm:diagonal}.
\end{remark}

\subsection{Stochastic operators with non-tensored support}\label{ssec:samp_stoch_op_nonrect}In this section we proceed to the most general case of stochastic operators. Consider an operator $\H$ with a spreading function $\steta(t,\nu)$ such that the autocorrelation function $\Reta(\tntn)$ is supported on some arbitrary bounded set $\SetM$ in $\R^4$.
We will see that in general it is no longer possible to reconstruct $\steta(t,\nu)$ itself.
For instance, this happens because with nonzero probability some instances $\steta(t,\nu; \w)$ may have spread larger than one, violating the necessary condition of the deterministic \autoref{thm:det}.
Nevertheless, one may hope to recover $\Reta(\tntn)$ from the autocorrelation of the received information $\Rf(t,t')$.

To prove \autoref{thm:stoch_nontensor_rectified} below, we will need to solve equations of the form $Y=GXG^*$ with both $X,Y$ \nnd.\footnote{
An hermitian matrix $X \in \C^{n\times n}$ is \emph{\nnd} if for any $a\in \C^n$, $a^*Xa\geq 0$.
}
To this end, we need a standard technique from linear algebra.
Let \emph{vectorization} $\vect \colon\C^{n\times n} \to \C^{n^2}$ be the linear isomorphism between the space of matrices $\C^{n\times n}$ and the space of column vectors $\C^{n^2}$ given by stacking of the columns
\[ \left(\vect A\right)_i = A_{i \bmod n, \floor{\nicefrac{i}{n}}}, \quad i=0,\dotsc, n^2-1. \]
The following identity relating vectorization and the Kronecker product of matrices is well known~\cite{VanLoan, TTT}:
\begin{equation}\label{eq:vectorization}
 \vect (AXB^t) = (A\otimes B)\vect X.
\end{equation}
For an arbitrary matrix $A$, the set $\Lambda_A = \{(\lambda,\lambda') \colon A_{\lambda,\lambda'} \neq 0\}$ is called \emph{the support set} of $A$. From the properties of \nnd matrices it is easy to see that
\begin{equation}\label{eq:admissible}
(\lambda,\lambda') \in \Lambda \text{ implies } \{ (\lambda',\lambda), (\lambda,\lambda), (\lambda',\lambda') \} \subseteq \Lambda.
\end{equation}
With some abuse of terminology, we call sets that may appear as support sets of \nnd matrices, and hence satisfy \eqref{eq:admissible}, \emph{\nnd patterns} or, for short, \emph{spd patterns}.

\begin{lemma}\label{lem:equivalence}
Let $\Lambda$ be a fixed finite \spd pattern, $\Lambda \subseteq \{(0,0),\dotsc, (n-1,n-1)\}$, and let $G \in \C^{n\times m}$. The following are equivalent:
\begin{enumerate}[(i)]
\item \label{item:i} For each \nnd $Y\in \C^{m\times m}$, there exists a unique $X\in \C^{n\times n}$ such that
			\begin{inparaenum}[a)]
			\item $X$ is \nnd,
			\item $\supp X = \Lambda$, and
			\item $Y = GXG^*$.
			\end{inparaenum}

\item \label{item:ii}
			If a Hermitian matrix $N\in \C^{n \times n}$ with $\supp N \subseteq \Lambda$ solves the homogeneous equation $0 = GNG^*$, then $N=0$.

\item \label{item:iii}
			The matrix $\GG\eval_\Lambda$ has a left inverse (is full rank).
\end{enumerate}
\end{lemma}
\begin{proof}

\eqref{item:i} $\Rightarrow$ \eqref{item:ii}
By contraposition, let there be $0\neq N\in \C^{n\times n}$ such that $GNG^*=0$; let $E$ be the diagonal matrix whose $\lambda$th diagonal entry is one if $(\lambda, \lambda)\in \Lambda$ and zero else.  Then there exists a positive real number $C$ (for example, Gershgorin circle theorem guarantees $C = \norm{N}_1 = \sup_{1\leq i \leq m} \sum_{j=1}^m \abs{N_{ij}} <\infty$ will be enough) such that both $C E + N$ and $C E $ are \nnd, and $G(C E +N)G^* = G(C E )G^* + GNG^* = G(C E )G^*$, thus violating the uniqueness of $X$ in \eqref{item:i}, as both $\supp (C E + N)$ and $\supp N \subseteq \Lambda$.

\eqref{item:ii} $\Rightarrow$ \eqref{item:i}
Suppose, that there exist $X_1\neq X_2$, both \nnd, supported on $\Lambda$, such that $GX_1G^* = GX_2G^*$. Then $N = X_1-X_2 \neq 0$ is necessarily Hermitian, $\supp N \subseteq \Lambda$, and $GNG^*=0$, a contradiction.

\eqref{item:iii} $\Rightarrow$ \eqref{item:ii}
Let an Hermitian matrix $N$ be supported on $\Lambda$ and be a solution to $0=GNG^*$. Applying vectorization to both sides of $0=GNG^*$, we get by \eqref{eq:vectorization}
\begin{equation}\label{eq:0=ggn}
\vect 0 = (\GG)\vect N,
\end{equation}
If $\GG\eval_\Lambda$ is invertible, \eqref{eq:0=ggn} implies $N=0$.

\eqref{item:ii} $\Rightarrow$ \eqref{item:iii}
Let $A\in\C^{m\times m}$ be an arbitrary square matrix supported on $\Lambda$ such that $\GG\eval_\Lambda \vect A = 0$, that is, $GAG^* = 0$. The \emph{Cartesian decomposition} of $A$ is given by $A  = H_1 + iH_2$ with both $H_1,  H_2$ Hermitian, defined by \[ H_1 = \frac12 (A+A^*), \quad H_2 = \frac{1}{2i}(A-A^*). \]
It is easy to see that both $\supp H_1, \supp H_2 \subseteq \Lambda$, and $GH_1G^*=GH_2G^* = 0$, since $GA^*G^* = (GAG^*)^* = 0$. Therefore, by \eqref{item:ii}, $H_1=H_2 =0$ , and $A=0$. It follows that $\vect A = 0$. Since $\vect A$ is arbitrary, it follows that the columns of $\GG\eval_\Lambda$ are linearly independent, that is, $\GG\eval_\Lambda$ is left invertible.
\end{proof}

\begin{definition}
We would say that a \spd support set $\Lambda$ is a \emph{permissible pattern} if some $c\in\C^L$ generates a Gabor frame $G = [ \MTc ]_{k,n=0}^{L-1}$ such that the equivalent conditions of \autoref{lem:equivalence} are satisfied. If an \spd pattern is not permissible, it will be called \emph{defective}.
\end{definition}
This designation reflects the emergence of these patterns as those which permit the sampling of operators with delta trains.

\begin{theorem}\label{thm:stoch_nontensor_rectified}
Let $\H\in\StOPW(\SetM)$ such that $\SetM$ is $(a,b,\Lambda)$-rectified for some $a,b>0$, $\abs{\Lambda} = L^2$, and $ab=1/L$.
If some $c\in\C^L$ generates $G = [ \MTc ]_{k,n=0}^{L-1}$ such that the submatrix $(\GG)\vert_\Lambda$ is (left) invertible, then we can reconstruct the autocorrelation of the spreading function $\E{\steta\, \conj{\steta}}$ from the autocorrelation $\Rf = \E{\f\,\conj{\f}}$ of the response $\f = \H \Shah_c$ of $\H$ to the $L$-periodic $c$-weighted delta train $\Shah_c = \sumZ{k} c_{k \bmod L} \, \delta_{a k}$ with the reconstruction formula
\begin{equation}\label{eq:reconstruction}
\begin{split}
 \Reta(\tntn)  &=  b^{-2}\sum_{j=0}^{L^2-1} \sum_{p,p'=0}^{L-1}\alpha_{j,p,p'}\epi{ab(n_jp-n'_jp')}  \\
&\qquad \times \summ \empi{a(\nu(k_j+mL+p) -\nu'(k'_j+m'L+p'))} \\
&\qquad \times \Rf (t-a(k_j+mL-p), \, t'-a(k'_j+m'L-p').
\end{split}
\end{equation}
\end{theorem}

\begin{proof}
Again, as in the deterministic case, we start with the support already rectified in the sense of \autoref{defn:M.is.rectified}.
We proceed in the same manner as in the direct product case until the mixing formula \eqref{eq:mixing_full}, which guarantees that for $p=0,\dotsc, L-1$ for all $(t,\nu)$ within the base rectangle $\rect = [0,a)\times[0,b)$
\[ \Zee(t,\nu) = \sum_{k,n\in\Z} \MTc \: \:\steta_{k,n}(t,\nu),\]
Taking autocorrelation on both sides, we get for all $p,p'=0,\dotsc,L-1$, and $(t,\nu),(t',\nu')\in\rect$,
\newcommand*{\miniMTc}[1][]{\Modul^{n#1} \Trans^{k#1} c}
\begin{align*}
 \E{\Zee_p(t,\nu)\, \conj{\Zee_{p'}(t',\nu')}} &= {\displaystyle \EXP\bigl\{\sumZ{k,n} ( \miniMTc )_p \steta_{k,n}(t,\nu) \!\sumZ{\kn} \conj{( \miniMTc['] )_{p'}\steta_{k'\!,n'}(t',\nu')}\bigr\} }\\
&=\sum_{\mathclap{\knkn\in\Z}}  (\MTc )_p \, \conj{(\MTc['])}_{p'} \E{\steta_{k,n}(t,\nu)\,\conj{\steta_{\kn}(t',\nu')}}.
\end{align*}
Let $\supp \Reta(\tntn)$ be covered by $L^2$ 4-dimensional parallelepipeds indexed by $\Lambda = \{(\bigknkn)\}$, that is,
\[
\Reta(\tntn) =
\sum_{\mathclap{(\knkn) \in \Lambda}} \E{\steta_{k,n}(t - ak,\nu-bn)\, \conj{\steta_{k'\!,n'}(t'-ak',\nu'-bn')}}.
\]
Without loss of generality, all $(k,n)$ and $(k',n')$ are distinct modulo $L$.\footnote{As in the deterministic case, the parameters $a,b,L$ can always be chosen in such a way that $\supp \Reta(\tntn) \subseteq [0,aL)\times[0,bL)$ up to a translation.
}
 For brevity, let $\mathcal{T}_j$ be the translation
\[ \mathcal{T}_j f(\tntn) = f(t-ak_j, \nu-bn_j, t'-ak'_j, \nu'-bn'_j). \]

Let us denote $Y_{p,p'} = \E{\Zee_p\,\Zee_{p'}^{*}}$ and $X_{\knkn} = \E{\steta_{k,n}\,\steta_{\kn}^{*}}$, where ${}^*$ stands for conjugate transpose. As covariance matrices of zero-mean random vectors, both  $X_{\knkn}$ and $Y_{p,p'}$ are \nnd, and conversely, every \nnd matrix is a covariance matrix for some random vector \cite{Papoulis}. With this notation, we write for each $(t,\nu), (t',\nu') \in \rect$ a deterministic matrix equation
\begin{equation}\label{eq:y=gxg}
Y = GXG^{*},
\end{equation}
where $G = [ \MTc ]$ is a Gabor matrix as in \eqref{eq:Gabor.frame}.

To have a a unique \nnd solution $X$ of the underdetermined system of equations \eqref{eq:y=gxg} with an arbitrary \nnd $Y$, by \autoref{lem:equivalence} it is necessary and sufficient that $\GG\eval_\Lambda$ is left invertible. Let $A$ denote its left inverse, \[
A  \coloneqq [\alpha_{(k,n),(p,p')}]_{k,n,p,p'=0}^{L-1} \text{ such that } A(\GG\eval_\Lambda) = \Id.
\]
Such $A$ may exist only if $\abs{\Lambda} \leq L^2$.

The autocorrelation of the channel spreading function can be recovered from $X$ by observing (starting with the definition \eqref{eq:Z=G_tilta} of $\steta_{k,n}$),
\begin{equation}\begin{split}\label{eq:13}
 \Reta(\tntn) &= \E{\steta(t,\nu)\,\conj{\steta(t',\nu')}}  \\ & = \sum_{j=0}^{L^2-1} \mathcal{T}_j X(\tntn)_{\knknj}\empi{a(\nu k_j - \nu'k'_j)} \\ & = \sum_{j=0}^{L^2-1} (A^{-1}\vect \mathcal{T}_j Y(\tntn)_{\knknj} \empi{a(\nu k_j-\nu'k'_j)}  \\ & = \sum_{j=0}^{L^2-1} \sum_{p,p'=0}^{L-1} \alpha_{j,p,p'}\mathcal{T}_j Y_{p,p'}(\tntn)\empi{a(\nu k_j -\nu'k'_j)},  \end{split}\end{equation}
where we construct $Y$ from the autocorrelation of the received signal $\f(t)=\H\Shah_c(t)$ via
\begin{align*}
 Y_{p,p'}(\tntn)  &= \E{\Zee_p(t,\nu)\, \Zee^*_{p'}(t',\nu')} \\ &= b^{-2} \empi{a(\nu p - \nu'p')} \, \E{\Zak\f(t+ap,\nu)\, \conj{\Zak\f(t'+ap',\nu')}} \\
&= b^{-2} \empi{a(\nu p - \nu'p')} \summ \empi{b^{-1}(m\nu-m'\nu')} \\
& \qquad \times  \E{\f(t+a(p-mL))\,\conj{\f(t'+a(p'-m'L))}}.
\end{align*}
Translating and simplifying, we get
\begin{align*}
\MoveEqLeft Y_{p,p'}(t-ak_j,\nu-bn_j; t'-ak'_j,\nu'-bn'_j)\\
&= b^{-2}\epi{ab(n_j p -n'_j p')}  \summ \empi{a((p+mL)\nu-(p'+m'L)\nu')} \\
&\qquad \times \Rf (t-a(k_j+mL-p), t'-a(k'_j+m'L-p).
\end{align*}
Upon substitution of the above back into \eqref{eq:13}, we get the desired reconstruction formula \eqref{eq:reconstruction}.
This completes the proof of \autoref{thm:stoch_nontensor_rectified}.
\end{proof}

If the support of $\Reta$ can be factored into a tensor square of some set of area less than one on the time-frequency plane, then \autoref{thm:stoch_nontensor_rectified} guarantees reconstruction of $\Reta$ from the \acorr $\Rf$ of the received measurement $\H \Shah_c$, which is a weaker result than \autoref{thm:stoch_tensor}. To wit, assume that $\Reta$ is supported on a product set $S\times S$ such that $S\times S$ is $(a,b,\Lambda)$-rectified, and that the area of $S$ is less than one. It follows that $\Lambda$ itself can be represented as a direct product of some index set, say, $\Lambda = \Gamma \times \Gamma$ such that $S$ is $(a,b,\Gamma)$-rectified.
For some vector $c$ the submatrix of its time-frequency shifts $G_{\Gamma}$ is left invertible \cite{PfaWal, LPW}.
Then $(\GG)\vert_\Lambda = \GGL{\Gamma}$ is also left invertible, and $\Reta(\tntn)$ can be recovered from $\Rf(t,t')$.

The strength of \autoref{thm:stoch_nontensor_rectified} is that it allows to exploit information about more complicated dependencies between scatterers manifested in the geometry of the support of the \acorr function $\Reta(\tntn)$.
\begin{figure}[H]
\centering
\goodcolor
\def\myscale{0.37}
\subfloat[Tensor square]{ \tikz[scale=\myscale] \tensorsupport; }
\subfloat[Arbitrary]{ \tikz[scale=\myscale] \curvysupport; } \\
\caption{Different types of support sets of autocorrelation functions.\label{fig:types.of.supports}}
\end{figure}
\begin{figure}[H]
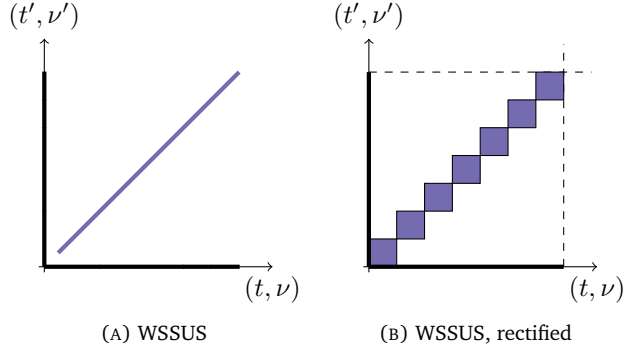

\centering
\goodcolor
\def\myscale{0.37}
\subfloat[WSSUS]{ \tikz[scale=\myscale] \wssussupport; }
\subfloat[WSSUS, rectified \label{subfig:wssus.rectified}]{ \tikz[scale=\myscale] \boxedwssussupport; }
\caption{Different types of support sets of autocorrelation functions, WSSUS case.\label{fig:wssus.supports}}
\end{figure}
For example, consider the important case of the WSSUS operators, already mentioned in the \autoref{rem:WSSUS1} after \autoref{thm:stoch_tensor}.
Let $\H^{\text{big}}$ be some WSSUS operator whose scattering function $C_\steta(t,\nu)$ is supported on a 2D rectangle $S^{\text{big}} \coloneqq [0,1)\times[0,L)$. This rectangle $S^{\text{big}}$ of area greater than one can always be rectified with a collection $\Lambda^{\text{big}}$ of $L^2$ translations of a box $[0,1)\times[0,\frac{1}{L})$.
From the WSSUS assumption we know that $\Reta(\tntn)=0$ whenever $(t,\nu)$ and $(t',\nu')$ are in different boxes, therefore, the 4D set $\supp \Reta(\tntn)$ is  $(1,\frac{1}{L},\Lambda^{\text{big}})$\Hyphdash* rectifiable, as shown in \autoref{subfig:wssus.rectified}.
It follows that in the sense of \autoref{thm:stoch_nontensor_rectified}, $\H^{\text{big}}$ is identifiable whenever the diagonal set $\Lambda^{\text{big}}$ \isPoU. Since, according to \autoref{thm:diagonal}, this is always the case with any such diagonal set, we obtain the following corollary of \autoref{thm:stoch_nontensor_rectified}.
\begin{theorem}\label{thm:wssus}
Let $\H\in \StOPW(S\times S)$ a WSSUS operator with  $S = \supp C(t,\nu)$ a compact set of arbitrary finite measure $\mu(S)<\infty$.
Then there exist $L$ prime, $L>\mu(S)$, $a,b>0$ with $ab=1/L$, a complex vector $c \in \C^L$ and a sequence $\{ \alpha_{jp} \}_{j,p=1}^L$ depending only on $c$ such that we can reconstruct the scattering function $C(t,\nu)$ of $\H$ from the autocorrelation of its response to an $L$-periodic $c$-weighted delta train $\Shah_c = \sumZ{k} c_{k \bmod L} \delta_{a k}$ using \eqref{eq:reconstruction.wssus}.
\end{theorem}
\begin{proof}
For an arbitrary $\eps$, it is always possible to find $a,b>0, \frac{1}{ab}>\sqrt{\mu(S)+\eps}$ such that the support of the scattering function $S$ is covered by $N$ translations of the rectangle $[0,a)\times[0,b)$ (as usual, indexed by $\Lambda$), such that $\mu(S)\leq Nab \leq \mu(S)+\eps$. Without loss of generality, all $(k,n)\in\Lambda$ are distinct modulo $L=\frac{1}{ab}$. We cover the 4D diagonal of the set $S\times S$ with $N$ 4D parallelepipeds of volume $1/(ab)^2$, so the total volume of the cover is $\frac{N}{L^2} \leq {1}$.
By \autoref{thm:stoch_nontensor_rectified}, there exists $c\in \C^L$ such that we can reconstruct the \acorr of the spreading function from the response $\f = \H \Shah_c$ using a version of  \eqref{eq:reconstruction} that takes advantage of the stationarity of $\H$ and periodicity of $\Shah_c$:
\begin{equation}\label{eq:reconstruction.wssus}
\begin{split}
C(t,\nu)  &=  b^{-2}\sum_{j=0}^{L^2-1} \sum_{p,p'=0}^{L-1}\sum_{m\in\Z} \alpha_{j,p,p'}\epi{a[(p-p')(\nu + bn_j) - \nu m L]}\\
&\quad \times \Rf (t-a(k_j-p), \, t-a(k_j+m'L-p').
\end{split}
\end{equation}
\end{proof}

\section{Permissible patterns}\label{sec:frame_theory}
The preceding discussion shows that it is ultimately important to study which \spd patterns $\Lambda$ are permissible, that is, for which sets $\Lambda$ there exists $c\in\C^L$ such that $\GG\eval_\Lambda$ is full rank.
Consider the spreading function $\steta(t,\nu)$ of $\H$ given by
\begin{equation*}\label{eq:steta_example}
\steta(t,\nu) = \sum_{k=0}^{L-1} \sum_{n=0}^{L-1} \steta_{k,n}(t - ka,\nu - nb). \end{equation*}
with $\steta_{k,n}(t,\nu)\in L^2(\R^2)$ supported on $[0,a)\times[0,b)$.
Due to the specifics of the underlying time-frequency analysis, it will be convenient to index the matrices with the finite index set $\Index$ given by
\begin{equation}\label{eq:Index}
{\textstyle \Index =  \left\{ (0,0), (0,1), \dotsc, (0,L-1), (1,0),\dotsc, (L-1,L-1)\right\}.}
\end{equation}

We call the \emph{autocorrelation pattern} of $\steta$ the indicator matrix $\ACP \in \{0,1\}^{\Index \times \Index}$
\begin{equation*}
\ACP(\lambda,\lambda') = \begin{cases}
1, & \text{ if for some }(t,\nu), (t'\!,\nu') \in \rect,  \E{\steta_{\lambda}(t,\nu) \, \conj{\steta_{\lambda'}(t',\nu')}} \neq 0 \\
0, &  \text{ otherwise,}
\end{cases}
\end{equation*}
where $\lambda = (k,n)$, $\lambda' = (k',n')$, $\lambda, \lambda' \in \Index$.
Clearly, $\ACP$ must be symmetric, moreover,  just as in \eqref{eq:admissible},
\begin{equation}\label{eq:acp}
\begin{split}
\ACP(\lambda,\lambda')=1 \quad \text{ implies } \quad  \ACP(\lambda',\lambda) = \ACP(\lambda,\lambda)=\ACP(\lambda',\lambda') = 1,
\end{split}
\end{equation}
and conversely, for any matrix $A\in \C^{\Index \times \Index}$ that satisfies \eqref{eq:acp}, there exists an operator with the spreading function $\steta_A(\tntn)$ whose \acorr pattern is $A$.
We visualize autocorrelation patterns with diagrams such as \autoref{fig:L=2}, where shaded boxes correspond to nonzero correlation between patches.

\autoref{lem:equivalence} indicates that the problem at hand is linked to the Haar property of Gabor systems.
We explore this connection in more detail in \autoref{ssec:rank_preserve}.
The Haar property of a finite-dimensional Gabor frame reads that any $L$ columns of a generic Gabor matrix $G$ are in general linear position. By the result from \cite{LPW}, this holds for almost every $G$ with $L$ prime.
Therefore, any deterministic spreading function supported on $L$ boxes of area $1/L$ on a time-frequency plane can be identified.
In other words, any set $\Gamma \subset \Index$ such that $\abs{\Gamma}=L$ indexes a subset of columns of $G$ that is linearly independent for almost every seed vector $c$.
This contrasts with the stochastic setting, where we will see that there exist plenty of \spd patterns that correspond to submatrices of $\GG$ which are rank-deficient for every choice of $c$.

\newcommand{\cc}{\tensoratom{c}{}{c}{}}

Generally, $\GG$ can be viewed as a Gabor system on the non-cyclic group $\Z_L\times \Z_L \cong  \Index $, generated by the window $\cc$,
\begin{multline*}
\GG = \{ \, \pi(\bigknkn) (\cc) \colon \\
(k,n),(k',n')\in \Index, \pi(\bigknkn) = \conj{\MTc } \otimes \, \Modul^{n'} \Trans^{k'}\!\!c \}.
\end{multline*}
The Haar property of $\GG$ would require all subsets of $\GG$ of order $L^2$ to be linearly independent, which is not the case for any prime number $L\geq 2$.
However, the positive-definiteness of the autocorrelation function demands the autocorrelation pattern to satisfy property \eqref{eq:acp}. This precludes certain subsets of $\GG$ from being tested for linear dependence.
Below we show that this restriction implies that for $L=2$ every \spd pattern is permissible; not so for $L\geq 3$.

\subsection{Case \texorpdfstring{$L=2$}{L=2}}\label{ssec:L=2}On Figures \ref{fig:L=2} and \ref{fig:L=2_other} we show all possible \spd patterns for $L=2$. In addition, the last pattern on \autoref{subfig:2offtopic}, corresponds to a linear dependent set of elements of $\GG$. However, due to lack of symmetry, such a set cannot appear as a support set for the \acorr of the spreading function of the stochastic operator.
We mention it here to highlight the connection of the stochastic operator identification theory with the analysis of Gabor frames on general abelian groups, for this example, $\Z_2\times \Z_2$. Further properties of such frames, including this pattern, can be found in~\cite{KPR08, PfaFinite}.

\begin{figure}[h]
\def\myscale{0.5}
\centering
\goodcolor
\subfloat
{\label{subfig:2tensor12} \tikz[scale=\myscale] \mysympattern{2}{1,2}{1/2}; }
\subfloat
{\label{subfig:2tensor13} \tikz[scale=\myscale] \mysympattern{2}{1,3}{1/3}; } \\
\subfloat
{\label{subfig:2tensor14} \tikz[scale=\myscale] \mysympattern{2}{1,4}{1/4}; }
\subfloat
{\label{subfig:2tensor23} \tikz[scale=\myscale] \mysympattern{2}{2,3}{2/3}; }\\
\subfloat
{\label{subfig:2tensor24} \tikz[scale=\myscale] \mysympattern{2}{2,4}{2/4}; }
\subfloat
{\label{subfig:2tensor34} \tikz[scale=\myscale] \mysympattern{2}{3,4}{3/4}; }
\caption{Tensor rank-1 \acorr patterns for $L=2$. Any rank-1 tensor pattern is permissible by \eqref{eq:rank_tensor}.\label{fig:L=2}}
\subfloat[]
{\label{subfig:2diag} \tikz[scale=\myscale] \mysympattern{2}{1,...,4}{}; }
\subfloat[]
{\label{subfig:2offtopic} \badcolor \tikz[scale=\myscale] \mynonsympattern{2}{1,3}{1/2,4/3}; }
\caption{Other \acorr patterns for $L=2$. \label{fig:L=2_other}}
\end{figure}
Interestingly, the only admissible patterns on $L=2$ belong to one of the two types, shown on \autoref{fig:L=2} and \autoref{subfig:2diag}. The first type contains sets that can be represented as rank-1 products\footnote{We say that an element $b \in S\otimes S$ has \emph{tensor rank} $k$ if $k$ is the smallest number with the property that there exist $\lst{a}_k, \lst{b}_k \in S$ such that $\Lambda = a_1 \otimes b_1 + \dotsb + a_k \otimes b_k$.} of sets of order $L$ in $\Index$, that is, there exists a set $\Gamma$ or order $\abs{\Gamma} = L$ such that $\Lambda = \Gamma\times \Gamma$. Such product patterns inherit the linear independence properties of their factors, since the Haar property of $G$ guarantees
\begin{equation}\label{eq:rank_tensor}
\rank \bigl(  \GG\eval_{\Gamma\times \Gamma} \bigr) = (\rank G_{\Gamma})^2 = L^2.
\end{equation}
The second type, all boxes on the diagonal, of maximal possible tensor rank, is also a permissible pattern for $L=2$, as we observe the determinant of the matrix
\begin{equation*}\label{eq:det.eq.2}
\det  \bigl( \GG\eval_{\text{diag}} \bigr)= 4(\abs{c_0}^4 - \abs{c_1}^4)(c_0^2\,\conj{c_1}^2 - c_1^2\,\conj{c_0}^2) \neq 0
\end{equation*}
or almost all choices of the generating atom $c = [c_0 \  c_1]^t$.

This concludes the analysis of the $L=2$ case with a happy ending:
\begin{proposition}
Let a function $\Reta(\tntn)$ be supported on a set of measure less than or equal to one covered with 4 boxes $\rect^2+(k_ja,n_jb, k'_ja, n'_jb)$ for some $\{(k_j,n_j, k'_j, n'_j) \in \Index\times\Index\}_{j=0}^3$, and some $\rect^2=[0,a)\times[0,b)\times[0,a)\times[0,b)$ such that $ab = \nicefrac{1}{2}$ and all $(k,n), (k',n')$  distinct modulo $2$. Let $\steta(t,\nu)$ be some function with $\Reta(\tntn)$ as its \acorr function, and $\H$ a stochastic channel with $\eta$ its spreading function.
Then we can recover $\Reta(\tntn)$ from the response of $\H$ to the sounding signal
\[
\Shah_c(t) = c_0\sum_m \delta(t-a - 2am) + c_1\sum_m \delta(t-2am)
\]
for almost all $c = [c_0 \ c_1]^t \in \C^2$, $\abs{c_0}=\abs{c_1}=1$.
\end{proposition}

\subsection{Case \texorpdfstring{$L>2$}{L>2}}\label{ssec:L=2_and_up}Classification of all \acorr patterns with $L=3$ already presents some challenges. For once, there are over 5000 possible \spd patterns
\[
1 + \binom97\binom72 + \binom95\left(2\binom54+3\binom53\right) = 5796.
\]
Here, 1 corresponds to the diagonal case, plus there are $\binom97$ choices of 7 boxes on the diagonal, any two of which can be correlated corresponding to the $\binom72$ factor, etc.\footnote{The case of even number of boxes on the diagonal is always counted within the larger odd case; by symmetry considerations, every correlation between boxes adds two boxes on the field). }

The problem of classifying all possible arrangements of boxes becomes intractable very quickly with $L$.
In \autoref{ssec:rank_preserve}, we show that whenever some $\Lambda$ is a permissible pattern for almost all $c$, there is a related (via the permutations of the Gabor elements) family of $\Lambda^\perm$ that will be permissible patterns for almost all $c$.
Such observations simplify the classification problem, although they do not suffice to give a complete classification.

Secondly, alas, unlike $L=2$, some \spd patterns correspond to rank-deficient subsets of $\GG$. That is, there are sets that satisfy \eqref{eq:admissible} but are not permissible, for example, see \autoref{fig:L3bad}. It turns out that there are structural reasons for patterns on \autoref{fig:subfig3} and \autoref{fig:L3bad} to be permissible (respectively, defective).

With little effort it can be seen that patterns with tensor rank 2 (as on \autoref{fig:L3bad}) can never be permissible, owing solely to the properties of the tensor product space and not to the specifics of the underlying Gabor structure. In fact, for any $L$,  permissible patterns cannot contain two complete large sets of pairwise tensor products.
(Here, such two sets are $\{\steta_{00},\steta_{01}\}$ and $\{\steta_{02}, \steta_{10}\}$, and in this context, \emph{large} means that $\card{\steta_{00},\steta_{01}} + \card{\steta_{02}, \steta_{10}} = 2+2>3$.)

\newcommand*{\ggii}{\tensoratom{g}{i}{g}{i'}}
\newcommand*{\ggjj}{\tensoratom{g}{j}{g}{j'}}
\newcommand*{\bb}[1]{\tensoratom{b}{#1}{b}{#1}}
\begin{proposition}\label{thm:two_squares}
Let $G\in \C^{n\times m}$ be an arbitrary matrix, and $G_1, G_2$ its submatrices comprising columns of $G$ indexed, respectively, by $\Gamma_1$ and $\Gamma_2$ such that $\Gamma_1 \cap \Gamma_2 = \varnothing$. Denote
$
\Lambda = (\Gamma_1 \times \Gamma_1) \cup  (\Gamma_2 \times \Gamma_2).
$
If $\abs{\Gamma_1} + \abs{\Gamma_2} > \rank G$, then  the set of elements in $\tensoratom{G}{}{G}{}\eval_\Lambda$ is linearly dependent.
\end{proposition}
\begin{proof}
Since $\abs{\Gamma_1} + \abs{\Gamma_2} > \rank G$, it follows that there exists $[b_1 \ b_2 ]^t \in \ker [ G_1 | G_2 ] \neq 0$, that is,
\begin{equation*}\label{eq:a.in.ker}
G_1 b_1 + G_2 b_2 = 0.
\end{equation*}
Now consider $B_1 = \bb{1} , B_2 = \bb{2}$. Both are Hermitian, and
\begin{align*}
G_1 B_1 G_1^* - G_2 B_2 G_2^* &= (G_1 b_1)(G_1 b_1)^* - (G_2 b_2)(G_2 b_2)^* \\
&= (G_1 b_1)(G_1 b_1)^* - (-G_1 b_1)(-G_1 b_1)^* = 0.
\end{align*}
This means that we have found a Hermitian matrix
\[
N = \begin{bmatrix} B_1 & 0 & 0 \\ 0 & B_2 & 0\\  0 & 0 & 0 \end{bmatrix}\
\]
 supported on $\Lambda$ such that $GNG^*=0$ and, hence, $\tensoratom{G}{}{G}{}\eval_\Lambda$ is linearly dependent. It follows that $\Lambda$ is not a permissible pattern.
\end{proof}
\begin{figure}[t]
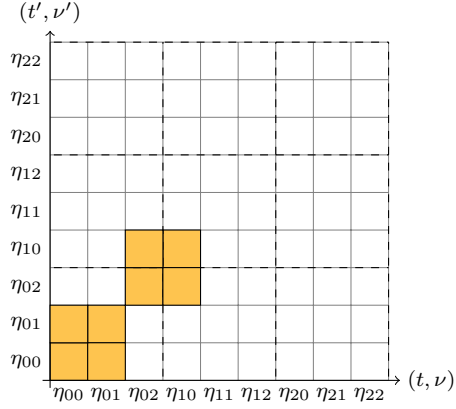

\centering
\def\myscale{0.5}
\badcolor
\tikz[scale=\myscale] \mysympattern{3}{1,...,4}{1/2,3/4}; \caption{According to \autoref{thm:two_squares}, any rank-2 pattern is not a permissible pattern.\label{fig:L3bad}}
\end{figure}

On the other hand, the patterns with maximal tensor rank (that is, diagonal patterns) will \emph{always} be permissible. This corresponds to the important case of WSSUS operators.
We give here an elementary proof for finite dimensions. (See also \cite[Theorem 3.1]{BenPfa} for a discussion of this in the realm of Hilbert-Schmidt operators in infinite dimensions.)
\newcommand{\GGdiag}{\GG\eval_{\text{diag}}}
\newcommand{\MT}{\Modul^n \Trans^k}
\begin{theorem}\label{thm:diagonal}
Denote $\pi(k,n) \coloneqq \MT$. Let $G = \{\pi(k,n) c\}_{k,n=0}^{L-1}$  be a Gabor frame generated by $c$. Then for almost all $c \in \C^L$, the set $\GGdiag$ of all tensor products of each Gabor element with itself
\[ \GGdiag \coloneqq \{ \tensoratom{\pi(k,n)c}{}{\pi(k,n) c}{} \}_{k,n=0}^{L-1}
\]
is linearly independent, that is, for almost all $c\in \C^L$ the diagonal set
\[ \Lambda_{\text{diag}} = \{ (k,n),(k,n) \}_{k,n=0}^{L-1} \]
is a permissible pattern.
\end{theorem}
\begin{proof}
It is a straightforward observation that the linear independence of the set $\GGdiag\in \C^{\Index\times \Index}$ is equivalent to the linear independence of the family of the corresponding finite-dimensional Hilbert-Schmidt operators $P_{k,n} \colon \C^L\to\C^L$ given by
\[
P_{k,n}x = \ip{x,\pi(k,n)c}\,\pi(k,n)c.
\]
Clearly, $P_{k,n} = (\MT) P_{0,0} (\MT)^*$.
From the spreading function representation of $H \colon \C^L\to \C^L$, \[
H = \sum_{p,q=0}^{N-1} \eta[p,q] \Modul^p \Trans^q
\]
it is easy to see with the help of the commutation relation $\Trans^k\Modul^n = \empi{kn}\MT$ that for any $H$ we have \begin{align*}
(\MT) H (\MT)^*
=\sum_{p,q=0}^{L-1} (\Modul^{(k, n)} \eta[p,q] ) \Modul^p \Trans^q,
\end{align*}
where $\Modul^{(k,n)}\eta[p,q] \coloneqq \epi{(qk-pn)}\eta[p,q]$.
In particular, the spreading functions of $P_{k,n}$ and $P_{0,0}$ satisfy
$\eta_{k,n} = \Modul^{(k,n)} \eta_{0,0}$.
Therefore, the linear independence of the family of operators $\{P_{k,n}\}_{k,n=0}^{L-1}\subseteq \HS(\C^L)$ is equivalent to the linear independence of the family $\{\Modul^{(k,n)}\eta_{0,0}\}\subseteq \C^{L^2}$.

In finite dimensions, the short time Fourier transform of $c$ with respect to itself takes form
\[ V_c c[p,q] = \ip{c, \Modul^p \Trans^q \, c} = \sum_r \empi{p r/L} \, c[r]\,\conj{c}[r-q].\]
Here and in what follows, the summation is over $\Z_L$ and all the indices of vectors $c,x \in \C^L$ are taken modulo $L$.
We compute
\begin{align*}
\allowdisplaybreaks
P_{0,0}x[m] &= \ip{x, c} \, c[m]  = \sum_q x[q]\, \conj{c}[q] \, c[m] \\
        &=\sum_q x[m-q] \,\conj{c[m-q]} \, c[m] \displaybreak[3]\\
        &=\sum_q x[m-q] \, \sum_r \delta[r-m] \, \conj{c[r-q]} \, c[r] \displaybreak[3]\\
         &= \sum_q  x[m-q] \, \sum_r \left( \frac{1}{L} \sum_p \epi{p(r-m)/L} \right) \conj{c}[r-q] \, c[r] \displaybreak[3]\\
        &= \sum_p \sum_q \bigl( \frac{1}{L} \sum_r \empi{p r/L} c[r] \, \conj{c}[r-q] \bigr) \epi{m p/L} \, x[m-q] \displaybreak[3]\\
        &= \sum_p \sum_q \left(\frac{1}{L}V_c c[p,q]\right)  \Modul^p \Trans^q x[m],
\end{align*}
which proves that $\eta_{0,0}[p,q] = \frac{1}{L} V_c c[p,q]$.

We claim that for almost every $c\in \C^L$, $V_c c$ has full support. For a fixed pair $p,q\in \Z_L\times\Z_L$,  the set of solutions $Z_{p,q} = \{c\in \C^L: V_c c\,[p,q]=0\}$ is a manifold of zero measure in $\C^L$.  Hence, $V_c c$ has full support almost everywhere, namely whenever $c \in \C^{L} \setminus \bigcup_{p,q} Z_{p,q}$.  The proof is complete by observing that whenever $\eta_{0,0}$ has full support, the family $\{\Modul^{(k,n)}\eta_{0,0}\}_{(k,n)\in\Z_L\times\Z_L}$ is linearly independent in $\C^{L^2}$.
\end{proof}
\begin{figure}[ht]
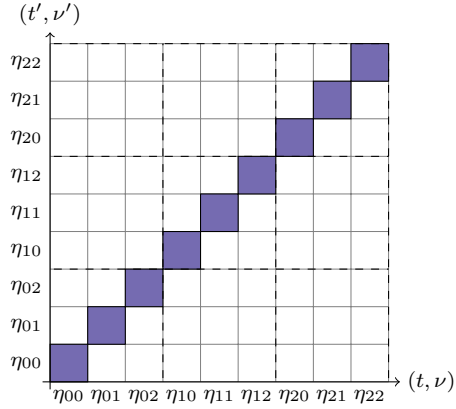

\centering
\def\myscale{0.5}
\goodcolor
\tikz[scale=\myscale] \mysympattern{3}{1,...,9}{};
\caption{By \autoref{thm:diagonal}, all diagonal patterns are permissible for almost all $c$.\label{fig:subfig3}}
\end{figure}
\begin{figure}[ht]
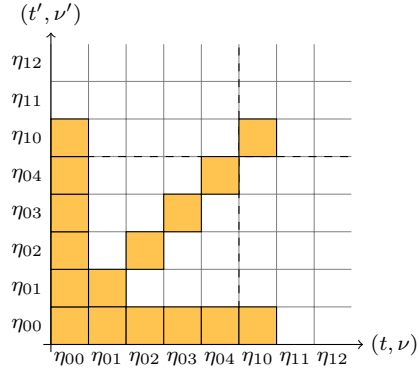

\centering
\badcolor  \tikz[scale=0.5] \mysympatternwithdots{5}{8}{1,...,6}{1/2,1/3,1/4,1/5,1/6}{7/{}/{},8/{\ldots}/{\vdots\;}};
\caption{\label{fig:arrowhead} Defective autocorrelation pattern on $L=5$. Total area $16/25$ leads to a $16 \times 25$ matrix of rank 13.}
\end{figure}
Another source of deficiency that begins to appear only with $L\geq 4$ is illustrated on \autoref{fig:arrowhead}.
\begin{proposition}\label{thm:tall}
Let the \emph{height} (or, equivalently, the width) of the pattern $\Lambda$ exceed $L$, or, formally, for some $\lambda_0\in\Lambda$,
\begin{equation*}\label{eq:lambda_o}
\#\{ \lambda' \colon (\lambda_0, \lambda') \in \Lambda \} > L.
\end{equation*}
Then $\Lambda$ is not the permissible pattern. The matrix $\GG\eval_\Lambda$ is singular for all $c\in \C^L$.
\end{proposition}

\begin{proof}
Let $\Gamma = \{ \lambda' \colon (\lambda_0, \lambda') \in \Lambda \}$.
Then $\abs{\Gamma} > L$, therefore, the set $G\eval_\Gamma$ is linear dependent.
Hence, we can find a nontrivial vector $a\in \ker G\eval_\Gamma$.
We now have $(\GG)a= \tensoratom{G}{}{Ga}{} = 0$.
That is, we have found a nontrivial linear combination $\sum_{i\in\Gamma} a_i g_i = 0$, and
\[ \sum_{i \in \Gamma} \tensoratom{a_i g_i}{}{g}{1} = \sum_{i\in\Gamma} a_i \left(\tensoratom{g}{i}{g}{1}\right) = 0. \]

The theorem immediately follows from \autoref{lem:equivalence} by observing that whenever there exists a non-trivial $A$ supported on $\Lambda$, not skew-Hermitian and such that $GA=0$, (here, $A= \bigl[ a \  0 \bigr]$), then a non-trivial Hermitian matrix $A + A^*$ satisfies $G(A+A^*)G^* = 0$.
\end{proof}

The existence of defective \spd sets shows that one must be careful when trying to reconstruct the \acorr of the spreading function by sampling an operator with the delta train.
Luckily, the task of weeding out defective patterns is uncoupled from the sampling procedure.
All defective patterns $\Lambda^{\text{bad}}$ can be discovered numerically and inexpensively by testing a Gabor matrix generated by a random vector $c$.
It is easy to see that the rank deficiency of $\GG\eval_{\Lambda}$ will (with probability one) indicate whether $\Lambda$ is good or bad.

\begin{remark}
\autoref{thm:stoch_nontensor_rectified} gives a procedure to recover $\E{\steta\, \conj{\steta}}$ from the autocorrelation $\E{\f\conj{\f}}$ of the response $\f = \H\Shah_c$ of $\H$ to an $L$-periodic $c$-weighted delta train $\Shah_c= \sum_{k\in\Z} c_{k \bmod L} \, \delta_{a k}$.
The result is based on finding a permissible rectification of the set $\SetM$. If some rectification satisfies the symmetry conditions for autocorrelations, but is not permissible, one can seek alternative (possibly finer) rectifications with the hope that one of them is permissible.  However, Propositions~\ref{thm:two_squares} and~\ref{thm:tall} indicate that for some sets $\SetM$ there are no permissible rectifications, hence, \autoref{thm:stoch_nontensor_rectified} does not apply to $\StOPW(\SetM)$.

Indeed, if for some $(t_{0},\nu_{0})$, we have $\mu(\EXP\{\steta(t_{0},\nu_{0})\, \conj{\steta(t',\nu')\}})>1$, then every rectification will have the defect discussed in \autoref{thm:tall}, and our procedure to recover $\Reta(\tntn)$ is not applicable.

Similarly, if for two sets $S_{1},S_{2}$ we have $\mu(S_{1})+\mu(S_{2})>1$ and $(S_{1}\times S_{1})\cup (S_{2}\times S_{2})\subseteq \SetM$.
Consider $\SetM$ be $(a,b,\Lambda)$-symmetrically rectified, and $S'_1$ be the induced $(a,b,\Gamma_1)$\Hyphdash* rectification of $S_1$, and $S'_2$ be the induced $(a,b,\Gamma_2)$\Hyphdash* rectification of $S_2\setminus S'_1$.
(Note that $\Gamma_2$ is not empty unless $S_2 \subset S_1'$, but then $\mu(S'_1)>1$, and the rectification is again defective by \autoref{thm:tall}).
Clearly, $\Gamma_1 \cap \Gamma_2 = \emptyset$, $(\Gamma_1 \times \Gamma_1) \cup (\Gamma_2 \times \Gamma_2) \subset \Lambda$, and $\mu(S'_1)+\mu(S'_2) \geq \mu(S_1) + \mu(S_2) > 1,$
and \autoref{thm:two_squares} applies.

This observation does not imply that there exist symmetric $\SetM$ with 4D volume less than one and the property that no test signal $g$ allows to recover $\E{\steta\, \conj{\steta}}$ from the autocorrelation $\E{\H g\,\conj{\H g}}$.  In fact, recovery may be possible using a different type of a pilot signal, for example, a non-periodic delta train.
\end{remark}

\subsection{Equivalence classes of permissible patterns}  
\label{ssec:rank_preserve}
Let $\Sym$ be the \emph{symmetry group} of all permutations on the group $\Index$ of order $L^2$ defined in \eqref{eq:Index}.
\begin{definition}\label{defn:equiv}
We call two \spd patterns $\Lambda$ and $\Lambda^\perm$ on the grid $\Index\times\Index$ \emph{homologous} whenever there exists a permutation $\perm\in \Sym$ such that
\[  (\lambda, \lambda') \in \Lambda \quad \Leftrightarrow \quad  (\perm(\lambda), \perm(\lambda'))\in \Lambda^\perm.\]
Two homologous patterns $\Lambda$ and $\Lambda^\perm$ are \emph{equivalent} if for almost all choices of the window $c\in \C^L$ generating $G$,
\[
\rank \GG\eval_\Lambda = \rank\GG\eval_{\Lambda^\perm}. \]
\end{definition}
For example, the patterns on \autoref{fig:L=3_other} are all homologous, with the first two provably equivalent, according to \autoref{thm:equivalent} below.
\begin{figure}[t]
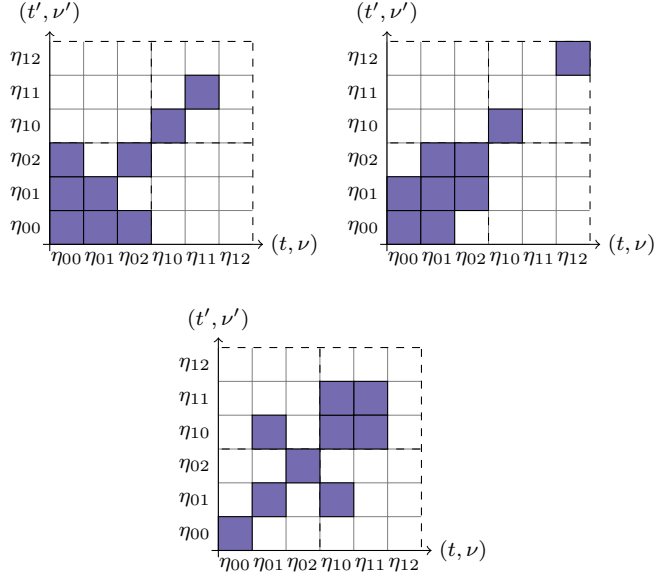

\centering
\def\myscale{0.45}
\goodcolor
\subfloat
{\label{fig:3fish1} \tikz[scale=\myscale] \mysympatternwithdots{3}{6}{1,...,5}{1/2,1/3}{4/{}/{},5/{}/{},6/{\ldots}/{\vdots\;}}; }
\subfloat
{\label{fig:3fish2} \tikz[scale=\myscale] \mysympatternwithdots{3}{6}{1,...,4,6}{1/2,2/3}{4/{}/{},5/{\eta_5}/{\eta_5},6/{}/{}}; }\\
\subfloat
{\label{fig:3fish3} \tikz[scale=\myscale] \mysympatternwithdots{3}{6}{1,...,5}{2/4,4/5}{4/{}/{},5/{}/{},6/{\ldots}/{\vdots\;}}; }
\caption{Three homologous patterns for $L=3$.\label{fig:L=3_other}}
\end{figure}

Numerical evidence shows that linear independence of the subsets of the tensored Gabor frame $\GG$ is invariant under homology in the sense of \autoref{defn:equiv}, for example, all the patterns on \autoref{fig:L=3_other} are permissible, as tested on a large number of randomly chosen $c$. Also, for $L=3,5$, and $7$, entire orbits of patterns $\{\Lambda^\perm\}_{\perm \in \Sym}$ are permissible. 
Additional indirect support to the hypothesis that any two patterns that are homologous with respect to the permutation of the Gabor atoms have the same linear independence status for almost all $c$ comes from the equivalence of  all patterns for the case $L=2$ (see \autoref{fig:L=2}) and the universality of  Theorems~\ref{thm:two_squares},~\ref{thm:diagonal} and~\ref{thm:tall} with respect to the permutations of the Gabor atoms.

We will describe the permutations $\perm \in \Sym$ such that $\Lambda$ and $\Lambda^\perm$ are provably equivalent for any $\Lambda$ for almost all $c\in \C^L$.
Clearly, all such permutations form a subgroup of $\Sym$. It remains an open question whether this subgroup is the whole group~$\Sym$.

For sake of notational simplicity, let us identify the atoms of a Gabor frame $G$ generated by a window $c\in \C^L$ with the elements of the torus $\Z_L\times \Z_L$. The $(k,n)$th atom $g = \MTc = \pi(k,n) \, c$ would correspond to $(k,n) \in \Z_L\times \Z_L$, where finite-dimensional operators $\Modul$ and $\Trans$ are defined in \eqref{eq:define.fin.dim.TM}, and $\pi(k,n) = \Modul^n \Trans^k$.
\begin{proposition}\label{thm:equivalent}
Let $A$ denote the group of permutations $\perm \in \Sym$ such that $\Lambda$ and $\Lambda^\perm$ are equivalent for all \spd patterns $\Lambda \subseteq \Index$ for all $c\in\C^L$, except maybe for a set of measure zero.
\begin{enumerate}[a)]
\item \label{item:trans} translations of the torus $\perm_{a}(k,n) \coloneqq (k+q, n+p) \bmod L$ belong to $A$ for all $p, q \in \Z$.
\item reflection of the torus $\perm_{b}(k,n) \coloneqq (-n, k) \bmod{L}$ belongs to $A$.
\item rotation of the torus $\perm_{c}(k,n) \coloneqq (k,-n)\bmod L$ belongs to $A$.
\end{enumerate}
\end{proposition}
For an illustration of these transformations in the case of $\Z_4\times \Z_4$, see Figures~\ref{fig:torus3ops} and~\ref{fig:dft}.
\begin{proof}
Note that the linear independence of the subset $G_\Lambda$ of elements of $G$ is invariant under scaling the atoms by scalars $p_\lambda$, complex conjugation, and any invertible linear transformation $T$. Let $G = [\MTc]_{k,n=0}^{L-1}$. First, we construct several Gabor frames closely related to $G$.

\begin{enumerate}[a)]
\item It is easy to see that the Gabor frame $G_a$ constructed from the atom $c_a = \phi_a(c) \coloneqq \Modul^p \Trans^q \, c$ is a column permutation of $G$ up to constant multiples of columns by $C = \epi{p k}$. Indeed,
\begin{align*} \pi(k,n)\, c_a & = \Modul^n \Trans^k (\Modul^p \Trans^q c) \\
&= \Modul^n \epi{p k} \Modul^p \Trans^k \Trans^q c \\
&= \epi{p k} \Modul^{n+p}\Trans^{k+q} c \\
&= C_a \pi(k+q,n+p) c.
\end{align*}
\begin{figure}[H]
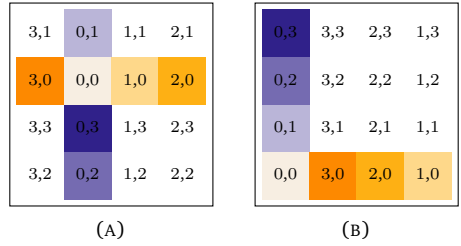

\centering
\subfloat[]{\label{fig:torustrans} \tikz[scale=0.9]	\torusdraw{mod(\x+4-1,4)}{mod(\y+4-2,4)}{3}; }  \quad
\subfloat[] {\label{fig:torusrefl} \tikz[scale=0.9]	 \torusdraw{mod(4-\x,4)}{\y}{3}; }
\caption{Actions of the time-frequency shift (a) and conjugation (b) are translation and horizontal reflection, respectively.\label{fig:torus3ops} }
\end{figure}
\item Similarly, constructing the Gabor frame $G_b$ with the Gabor window being the Fourier transform of the atom $c_b = \phi_b(c) \coloneqq \FT c$ produces a Fourier image of the original Gabor matrix. This holds since the Fourier transform commutes with time-frequency shifts in the following sense:
\begin{align*}\label{eq:FT}
\pi(k,n)c_b &= \Modul^n \Trans^k  \FT c \\
&= \FT(\Trans^{-n} \Modul^k c) \\
&=\FT(\epi{k n} \Modul^k \Trans^{-n} c) \\
&= C_b \, \FT \pi(-n,k) c.
\end{align*}
The Fourier transform is a unitary operation which acts on the Gabor matrix from the left, hence it does not change the linear independence of the columns, scaling of columns by the constant  $C_b$ notwithstanding.
In other words, $G_b$ is a frame that is \emph{unitarily equivalent} to a particular permutation of $G$, namely, it corresponds to rotating the torus, see \autoref{fig:dft}~\cite{HanLarson, ValeWaldron}.

Curiously, the vertical flip operation $\mathcal{I}c_j \coloneqq c_{-j}$ also creates an equivalent pattern. This can be understood by observing $\mathcal{I} = \FT^2$.
This corresponds to the reflection of both rows and columns of the torus $\Z_L \times \Z_L$.
\begin{figure}[h]
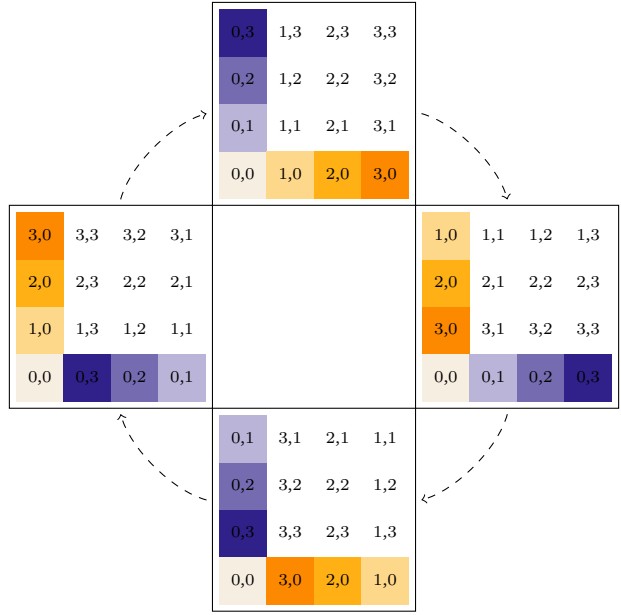

\centering
\tikz[scale=0.9]\toruscombine;
\caption{Action of the Fourier transform is $90^\circ$ rotation.\label{fig:dft}}
\end{figure}
\item Another operation that preserves linear independence is conjugation. Let the Gabor frame $G_c$ be generated by the window $ c_c = \phi_c(c) \coloneqq \conj{c}$. Then $\conj{G_c}$ contains all the atoms of $G$ in the order that corresponds to the horizontal reflection of the torus group on itself. Note that taking conjugation does not destroy the linear independence of the vectors.
\end{enumerate}
Now let $\Lambda$ be any \spd pattern. From the above constructions of the Gabor frames $G_\alpha, \alpha=a,b,c$, it is evident that for all generating atoms $c_\alpha \coloneqq \phi_\alpha (c) \in \C^L$, $\rank \tensoratom{G}{\alpha }{G}{\alpha }\eval_{\Lambda^{\perm_\alpha }} = \rank \GG\eval_\Lambda$.
 For any $\alpha=a,b,c$ the transformation $\phi_\alpha\colon c\mapsto c_\alpha$ is a bijective mapping $C^L\to \C^L$, therefore,
if $\Lambda$ is permissible for some $c$, then $\Lambda^\perm$ is also permissible (for some other $c$), and moreover, if $\Lambda$ is permissible for all $c \in \C^L$ except a zero set $Z$, then $\Lambda^\perm$ is also permissible for all $c \in \C^L$ except a possibly different zero set $Z^\perm$. And conversely, if $\Lambda$ is defective, then $\Lambda^\perm$ is also defective.
Since there are finitely many permutations $\perm$, the union of all bad sets $\bigcup_{\perm \in \Sym} Z^\perm$ has measure zero, and it follows that $\Lambda$ and $\Lambda^\perm$ are equivalent.
\end{proof}

 \section{Conclusion}\label{sec:conclusions}We show that a large class of stochastic operators can be reconstructed from the stochastic process resulting from the application of the operator to an appropriately designed deterministic test signal. Our results are based on combining recently developed sampling results for deterministic band limited operators \cite{KozPfa,PfaWal,Pfander} and the sampling theory of stochastic processes \cite{Lee,Papoulis}. 
The classes of stochastic operators considered here are characterized by support constraints on the autocorrelation of the stochastic spreading function of the operator. Our results do not necessitate the frequently used WSSUS condition, nor the so-called underspread condition on operators. 
Moreover, we show that in some cases when the operator cannot be fully determined from its action on a test signal, still, the second order statistics of the spreading function, that is, its autocorrelation can be determined from the operator output.
While in the deterministic case the possibility of operator reconstruction only depends on the support size of the spreading function, we show that in the stochastic case geometric considerations play a fundamental role for reconstruction.

\bibliographystyle{plainnat}
\bibliography{PfaZh02}
%\printbibliography 
\end{document}